\documentclass[acmsmall]{acmart}

\def\BibTeX{{\rm B\kern-.05em{\sc i\kern-.025em b}\kern-.08emT\kern-.1667em\lower.7ex\hbox{E}\kern-.125emX}}
    
\setcopyright{acmlicensed}
\acmJournal{JETC}
\acmYear{2020} \acmVolume{1} \acmNumber{1} \acmArticle{1} \acmMonth{1} \acmPrice{15.00}\acmDOI{10.1145/3446215}\copyrightyear{2019}

\usepackage{amsfonts}
\usepackage{subfig}
\usepackage{xcolor}
\usepackage{amsmath}
\usepackage{multicol}
\usepackage{multirow}
\usepackage{algorithm}
\usepackage{array}
\newcommand{\ignore}[1]{}
\usepackage{pifont}
\newcommand{\xmark}{\text{\ding{55}}}
\usepackage{colortbl}

\definecolor{darkred}{rgb}{0.5, 0, 0}
\definecolor{darkgreen}{rgb}{0, 0.5, 0}
\definecolor{darkblue}{rgb}{0,0,0.5}

\newlength{\saveparindent}
\newlength{\saveparskip}
\newcounter{ctr}

\newcounter{ectr}

\newcommand\ZZ{\ensuremath{\mathbb Z}}

\newcommand\DDD{\ensuremath{\mathcal D}}
\newcommand\EEE{\ensuremath{\mathcal E}}

\newcommand\KKK{\ensuremath{\mathcal K}}

\newcommand\MMM{\ensuremath{\mathcal M}}

\newcommand\XXX{\ensuremath{\mathcal X}}

\newcommand{\etal}{\mbox{\emph{et al.\ }}}
\def\eg{{\it e.g. }}
\newcommand{\ie}{\mbox{\emph{i.e.\ }}}

\newlength{\protowidth}

\newcommand{\namedref}[2]{\hyperref[#2]{#1~\ref*{#2}}}

\newcommand{\subfigureref}[2]{\hyperref[#1]{Figure~\ref*{#1}#2}}

\begin{document}

%
\title{Robust and Attack Resilient Logic Locking with a High Application-Level Impact}

%

\author{Yuntao Liu, Michael Zuzak, Yang Xie, Abhishek Chakraborty, Ankur Srivastava}
\affiliation{%
  \institution{\\ University of Maryland, College Park}
}



 




%

\renewcommand{\shortauthors}{Yuntao Liu et. al.}

%
\begin{abstract}
Logic locking is a hardware security technique aimed at protecting intellectual property (IP) against security threats in the IC supply chain, especially those posed by untrusted fabrication facilities. Such techniques incorporate additional locking circuitry within an IC that induces incorrect digital functionality when an incorrect verification key is provided by a user. The amount of error induced by an incorrect key is known as the \textbf{effectiveness} of the locking technique.
A family of attacks known as "SAT attacks" provide a strong mathematical formulation to find the correct key of locked circuits. In order to achieve high \textbf{SAT resilience} (\ie complexity of SAT attacks), many conventional logic locking schemes fail to inject sufficient error into the circuit when the key is incorrect. For example, in the case of SARLock and Anti-SAT, there are usually very few (or only one) input minterms that cause any error at the circuit output. 
The state-of-the-art {\emph stripped functionality logic locking (SFLL)} technique provides a wide spectrum of configurations which introduced a trade-off between {\bf SAT resilience} and {\bf effectiveness}.
In this work, we prove that such a trade-off is universal among all logic locking techniques.
In order to attain high effectiveness of locking without compromising SAT resilience, we propose a novel logic locking scheme, called Strong Anti-SAT (SAS).
In addition to SAT attacks, removal-based attacks are another popular kind of attack formulation against logic locking where the attacker tries to identify and remove the locking structure. Based on SAS, we also propose Robust SAS (RSAS) which is resilient to removal attacks and maintains the same \textit{SAT resilience} and \textit{effectiveness} as SAS. 
SAS and RSAS have the following significant improvements over existing techniques. 
(1) We prove that the {\it SAT resilience} of SAS and RSAS against SAT attack is not compromised by increase in {\it effectiveness}.
(2) In contrast to prior work which focused solely on the circuit-level locking impact, we integrate SAS-locked modules into an 80386 processor and show that SAS has a high application-level impact.
(3) Our experiments show that SAS and RSAS exhibit better SAT resilience than SFLL and their effectiveness is similar to SFLL. 
\end{abstract}

%
%


%

%

%
\maketitle

\section{Introduction}

Due to the increasing cost of maintaining IC foundries with advanced technology nodes, many chip designers have become fabless and outsource their fabrication to off-shore foundries.
However, such foundries are not under the designer's control which puts the security of the IC supply chain at risk.
{Untrusted foundries are capable of various malicious activities, among which ``overbuilding'' is the most concerning since the foundry can simply produce more copies of the IC and sell the extra copies for profit.}
Many design-for-trust techniques have been studied as countermeasures among which logic locking has been the most widely studied \cite{chakraborty2019keynote}.
A logic locked circuit requires a secret key input and the correct key is kept by the designer and not known to the foundry. The functionality of the circuit is correct only if the key is correct.
After the foundry manufactures the locked circuit and returns it to the designer, the correct key is applied to the circuit by connecting a tamper-proof memory containing the key to the key inputs. This process is called \emph{activation}.
Over the years, different types of logic locking mechanisms have been suggested. 
Initially, locking involved inserting XOR/XNOR gates in a synthesized design netlist~\cite{roy2008epic}.
Later, techniques based on VLSI testing principles have been outlined to improve logic locking schemes by manifesting high corruption at the output bits when an incorrect key is applied \cite{rajendran2012security,rajendran2015fault}.
{It is noteworthy that logic locking has evolved from an academic proposal to a practical solution. Sisejkovic \etal proposed Inter-Lock, a cross-core logic locking technique \cite{vsivsejkovic2019inter}, based on which a scalable logic locking solution for multi-module hardware designs was implemented it on a RISC-V processor \cite{vsivsejkovic2020scaling, vsivsejkovic2020secure}, which has already been made commercially available.}

The Boolean satisfiability-based attack, a.k.a. SAT attack \cite{subramanyan2015evaluating} was a game changer in the logic locking field. SAT provides a strong mathematical formulation to find the correct locking key of a logic locked IC which prunes out wrong keys in an iterative manner. In each iteration, an input (called the Distinguishing Input, or DI) is chosen by the SAT solver and all the wrong keys that corrupt the output of this DI are pruned out. All wrong keys are pruned out when no more DI can be found.
Point function (PF)-based logic locking, including SARLock~\cite{yasin2016sarlock} and Anti-SAT~\cite{xie2018anti}, force the number of SAT iterations to be exponential in the key size by pruning out only a very small number of wrong keys in each iteration.
However, PF-based locking schemes have a drawback that there are very few (or only one) input minterms whose output is incorrect for each wrong key. Hence the overall error rate of the locked circuit with a wrong key is very small. This disadvantage is captured by approximate SAT attacks such as AppSAT~\cite{shamsi2017appsat} and Double-DIP~\cite{shen2017double}. These attack schemes are able to find an \textit{approximate key (approx-key)} which makes the locked circuit behave correctly for most (but not all) of the input values.
Another kind of popular attack against logic locking schemes is removal attacks \cite{yasin2017removal}. In a removal attack, the attacker tries to find the logic locking module, remove it, and replace its output with a constant 0 or 1. The key step in this attack is to identify the output wire of the locking module. This can be achieved by structural analysis assisted by calculating the signal probability skew (SPS) of each wire \cite{yasin2017removal}. Locking techniques such as Anti-SAT \cite{xie2018anti} is most vulnerable to this type of attack since the correct functionality of the original circuit can be obtained by removing the Anti-SAT module and replacing its output with 0.
{Other types of attacks on logic locking have also been proposed, such as Hamming distance guided hill climbing attack \cite{plaza2015solving}, reduced-order binary decision diagram (ROBDD)-based attack \cite{massad2017logic}, machine learning based structural attacks ``SAIL'' \cite{chakraborty2018sail} and ``SWEEP'' \cite{alaql2019sweep}, and combined structural and functional attack ``SURF'' \cite{chakraborty2019surf}. These attacks are effective in recovering \textit{most} of the key bits correctly. However, these attacks mainly targeted pre-SAT logic locking schemes and do not guarantee to find a correct key. Hence, in this paper, we focus on SAT and removal based attacks.}

More recently, Yasin \etal proposed {\it stripped functionality logic locking (SFLL)} which allows the designer to select a set of protected input patterns that are affected by a large percentage of wrong keys while other input patterns are affected by very few wrong keys~\cite{yasin2017provably}. SFLL is not vulnerable to removal attack since the functionality of the original circuit for the protected input patterns has been modified in SFLL.
However, when the number of protected patterns increases, SAT attacks need significantly fewer iterations to find the correct key. 
{More details of SFLL and attack methods targeting SFLL are introduced in Section \ref{ssec:anti-sat}.}
Essentially, SFLL creates a fundamental trade-off between {\bf SAT resilience} (\ie SAT attack complexity) and {\bf effectiveness} (\ie the amount of error injected by a wrong key).
This trade-off is problematic. 
On the one hand, if only very few input patterns are protected, a wrong key may not inject enough error into the circuit and useful work may still be done using the chip, rendering locking {\bf ineffective}.
On the other hand, having more protected input patterns will compromise the circuit's {\bf SAT resilience}.
Moreover. as we move into the machine learning (ML) era, error-resilient applications are becoming increasingly relevant since most ML-based applications usually embody substantial amount of error resilience. Hence small amount of error in the hardware (introduced by incorrect keys and/or hardware simplification) may not necessarily impact the overall application accuracy. 
With SFLL, if we want to ensure a very high corruption at the hardware level (for wrong keys), the resiliency to SAT would inevitably reduce. Addressing this dilemma is the main theme of our paper.

We propose \emph{Strong Anti-SAT (SAS)} to address the challenges in achieving high effectiveness without sacrificing SAT resilience.
On one hand, SAS ensures that, given any wrong (including approximate) key, the error injected by locking circuitry will have significant application-level impact. On the other hand, SAS is provably resilient to SAT attacks (\ie requiring exponential time). Based on SAS, we also propose Robust SAS (RSAS), a variant of SAS that is not vulnerable to removal attacks and has the same \textit{SAT resilience} and \textit{effectiveness} as SAS. This makes RSAS a substantial improvement over the limitations posed by SAS. The contribution of this work is as follows.

\begin{enumerate}
    \item We prove the fundamental trade-off between {\it SAT resilience} and {\it effectiveness} which is applicable to any logic locking scheme.
    
    \item We demonstrate the inability of existing locking techniques to secure hardware running real-world workloads due to such a trade-off. We show that, when the longest combinational path (\ie the multiplier) in a 32-bit 80386 processor is locked using SFLL, the processor fails to simultaneously have high SAT complexity and high application-level impact on both PARSEC \cite{bienia2008parsec} and ML-based application benchmarks.

    \item We propose \emph{Strong Anti-SAT (SAS)} to address this challenge. In SAS, a set of input minterms that have higher impact on the applications are identified as \textit{critical minterms}.
    We design the locking infrastructure of SAS such that given a wrong key, the critical minterms are more likely to introduce error in the circuit and hence result in an application-level error. We also prove that the SAT complexity is exponential in the number of key bits and {does not deteriorate when the number of critical minterms increases}. This is a substantial improvement over SFLL.
    
    \item We also propose a removal attack resistant variant of SAS, called Robust SAS (RSAS). RSAS is designed such that it achieves the same \textit{SAT resilience} and \textit{effectiveness} levels as SAS and if the locking module of RSAS is removed, the remaining circuit will exhibit incorrect functionality for critical minterms.

    \item Experiment results show that, when locked using the same number of critical minterms, SAS and RSAS have higher \textit{SAT resilience} than SFLL and they have about the same level of effectiveness. In terms of area, power, and delay overhead, RSAS and SFLL have similar overheads in general.
\end{enumerate}

The rest of the paper is organized as follows. 
Sec. \ref{sec:sas_bg} introduces the background on SAT attack and existing logic locking schemes.
We show that SFLL's trade-off makes it incapable to secure real-world applications in Section~\ref{sec:attacks}.
We then mathematically prove that the trade-off applies to all logic locking schemes in Section~\ref{sec:trade-off}.
In Section~\ref{sec:countermeasure}, SAS's hardware structure is presented and its exponential SAT attack complexity is proved in theory. 
The removal attack resistant variant of SAS, \ie RSAS, is introduced in Section \ref{sec:RSAS}.
Section~\ref{sec:critical_minterms} describes the methodology to choose critical minterms.
Section~\ref{Sec_Exp} shows the experimental results which demonstrate that when the same set of critical minterms are selected by SAS, RSAS, and SFLL, SAS and RSAS achieve higher \textit{security} than SFLL while maintaining similar application-level effectiveness. Section~\ref{Sec_Conclusion} concludes the paper.

\section{Background} \label{sec:sas_bg}
\subsection{Attack Model}\label{ssec:attack_model}
Fig.~\ref{fig:attack_model} illustrates the threat model we consider which is consistent with the latest papers in the logic locking field~\cite{xie2018anti, shamsi2017appsat, liu2020strong, sengupta2020truly}.
The attacker can be either an untrusted foundry or an untrusted user who has the ability to reverse engineer the fabricated chip and obtain the locked gate-level netlist.
The attacker is considered to have the following resources:

\begin{figure*}[h]
    \centering
    \includegraphics[width=.7\textwidth]{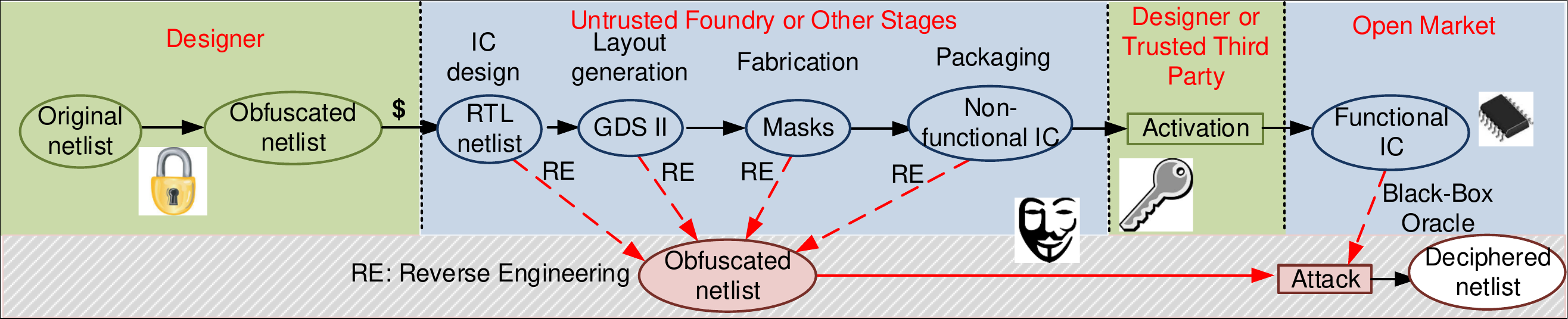}
    \caption{The targeted attack model of logic locking}
    \label{fig:attack_model}
\end{figure*}

\begin{enumerate}
    \item \textit{The locked gate-level netlist of the circuit under attack.} This can be obtained by reverse engineering the GDS-II file (which the foundry has) or a fabricated chip (which can be done by a capable end user).
    
    \item \textit{An activated chip.} 
    The attacker is considered to own an activated chip (\ie the one loaded with the correct key) since such a chip can be purchased from the open market. 
\end{enumerate}

In general, logic locking research does not assume that the attacker is able to insert probes into the activated circuit, \ie to observe the intermediate values. This is because protection schemes (\eg analog shield~\cite{ngo2017cryptographically}) can counter probing attacks.

\subsection{SAT Attack} \label{prelim}
{For any combinational digital circuit, the functionality can be expressed using a Boolean function $F: \vec{X}\rightarrow \vec{Y}$ where $\vec{X} \in \{0,1\}^n$ and $\vec{Y}\in \{0,1\}^{n_o}$ are the input and output of the circuit, respectively.
The logic locked circuit $F_L$ takes one more input, the key input $\vec{K}\in \{0,1\}^k$, in addition to the primary input $\vec{X}$, \ie $F_L: \vec{X},\vec{K}\rightarrow \vec{Y}$.}
If $\vec{K}$ is correct, then $\forall \vec{X}, F(\vec{X})=F_L(\vec{X},\vec{K})$.
$F(\vec{X})$ may not be equal to $F_L(\vec{X},\vec{K})$ if $\vec{K}$ is incorrect.
As stated earlier, the key is stored tamper-proof memory and is not accessible to the attacker.

The Boolean satisfiability-based attack, a.k.a. \emph{SAT attack} is a strong theoretical formulation to find the correct key of a locked circuit.
In the context of the SAT attack,  we use the \emph{Conjunctive Normal Form (CNF)}: $C(\vec{X},\vec{K},\vec{Y})$ to characterize Boolean satisfiability: $C(\vec{X},\vec{K},\vec{Y}) = \mathtt{TRUE}$ if  $\vec{X}$, $\vec{K}$, and $\vec{Y}$ satisfy $\vec{Y}=F_L(\vec{X},\vec{K})$, where $F_L$ stands for the Boolean functionality of the locked circuit. $C(\vec{X},\vec{K},\vec{Y}) = \mathtt{FALSE}$ otherwise.
SAT attacks run iteratively and prune out incorrect keys in every iteration. The attack consists of the following steps:
\begin{enumerate}
    \item In the initial iteration, the attacker looks for a primary input, $\vec{X}_1$, and two keys, $\vec{K}_\alpha$ and $\vec{K}_\beta$, such that the locked circuit produces two different outputs $\vec{Y}_\alpha$ and $\vec{Y}_\beta$:
        \begin{equation}
        C(\vec{X}_1,\vec{K}_\alpha,\vec{Y}_\alpha)\land C(\vec{X}_1,\vec{K}_\beta,\vec{Y}_\beta)\land (\vec{Y}_\alpha \ne \vec{Y}_\beta)
        \label{SAT_formula_first}
        \end{equation}
        $\vec{X}_1$ is called the \emph{Distinguishing Input (DI)}.
    \item The DI, $\vec{X}_1$, is applied to the activated circuit (the oracle) and the output $\vec{Y}_1$ is recorded.
        Note that $\vec{K}_\alpha$, $\vec{Y}_\alpha$, and $\vec{K}_\beta$, $\vec{Y}_\beta$ are not recorded.
        Only the DI and its correct output are carried over to the following iterations.
    \item In the $i^{\text{th}}$ iteration, a new DI and a pair of keys, $\vec{K}_\alpha$ and $\vec{K}_\beta$, are found.
    The newly found $\vec{K}_\alpha$ and $\vec{K}_\beta$ should produce correct outputs for all the DIs found in previous iterations. 
    To this end, we append a clause to Eq.~\eqref{SAT_formula_first}:
        \begin{equation}
        \begin{aligned}
        C(\vec{X}_i,\vec{K}_\alpha,\vec{Y}_\alpha)\land C(\vec{X}_i,\vec{K}_\beta,\vec{Y}_\beta)\land (\vec{Y}_\alpha \ne \vec{Y}_\beta) \\
        \bigwedge_{j=1}^{i-1}(C(\vec{X}_j,\vec{K}_\alpha,\vec{Y}_j)\land C(\vec{X}_j,\vec{K}_\beta,\vec{Y}_j))
        \end{aligned}
        \label{SAT_formula}
        \end{equation}
    In this way, all the wrong keys that corrupt the output of previously found DIs (\ie the output is different from that of the activated chip) are pruned out from the search space.
    
    \item SAT solves Eq. \eqref{SAT_formula} repeatedly until no more DI can be found, \ie Eq. \eqref{SAT_formula} is not satisfiable any more.
    \item In this case, there is no more DI. The output of the SAT attack is a key $\vec{K}$ that produces the same output as the activated circuit to all the DIs, which can be expressed using the following CNF:
        \begin{equation}
        \bigwedge_{i=1}^{\lambda}C(\vec{X}_i,\vec{K},\vec{Y}_i)
        \label{SAT_result}
        \end{equation}
        where $\lambda$ is the total number of SAT iterations.
\end{enumerate}

\begin{theorem}
SAT is guaranteed to find a correct key $\vec{K}_c$ to the locked circuit.
\label{thm:SAT}
\end{theorem}

The proof is given in Appendix~\ref{proof:thm_sat}. Note that there can be multiple correct keys: some keys can be different from but functionally equivalent to the actual key in the activated chip.


\subsection{Existing Logic Locking Schemes}\label{ssec:anti-sat}
Multiple logic locking schemes have been proposed to thwart the SAT attack~\cite{yasin2016improving, yasin2017provably, yasin2016sarlock}.
There are two ways to mitigate the SAT attack: one is to increase the time for each SAT iteration and the other is to increase the number of SAT iterations.
{The former requires adding SAT-hard circuitry such as AES blocks \cite{yasin2016improving} or permutation blocks \cite{kamali2019full}. These techniques usually incur huge area overhead which is impractical for most circuits.}
The other approach is to exponentially increase the number of SAT iterations. This approach is also not perfect because a locking scheme must be rather ineffective to improve security. This is the case for Anti-SAT~\cite{xie2018anti}, SARLock~\cite{yasin2016sarlock}, and and TTLock~\cite{yasin2017ttlock}. All these techniques are vulnerable to the approximate SAT attacks (such as AppSAT~\cite{shamsi2017appsat} and Double-DIP~\cite{shen2017double}).

The state-of-the-art {\it stripped functionality logic locking (SFLL)} \cite{yasin2017provably} explores the trade-off between SAT resilience and effectiveness. SFLL comprises of two parts: a functionality stripped circuit (FSC) and a restore unit (RU). The FSC is the original circuit with the functionality modified for a set of {\it protected input cubes}. This modification makes SFLL resistant to removal attack. If the RU is removed, the FSC's functionality of the protected input cubes is different from the original circuit, thus making the attack unsuccessful.
The RU stores the key, compares the circuit's input with the key, and outputs a {\it restore vector} which is XOR'ed with the FSC output. If the key is correct, the restore vector will fix the FSC's output and the circuit will have correct output.
{There are two types of SFLL: SFLL-HD and SFLL-flex. SFLL-HD leaves very specific structural traces in the FSC which have been successfully captured by a functional analysis based attack~\cite{sirone2020functional}. SFLL-flex, on the other hand, can leave almost no structural traces in the FSC if a fault-injection-based approach is taken to strip the functionality \cite{sengupta2020truly}. However, very recently, a sensitivity-based approach to identify the protected cubes of SFLL-flex is proposed \cite{sweeney2020sensitivity}. The sensitivity of a Boolean function on a specific input value measures how likely the output bit is to change when an input bit is flipped. The authors of \cite{sweeney2020sensitivity} found that in the ISCAS85 benchmark suite, circuits with more input bits tend to have lower sensitivity on average, based on which they proposed a SAT formulation to find input minterms with high sensitivities in the stripped-functionality circuit. In many benchmarks, the stripped input minterms can be found in this way. 
In terms of a countermeasure, the authors of \cite{sweeney2020sensitivity} proposed to find and strip the functionality of critical minterms which, after inverting their functionality, will have sensitivity values close to the average sensitivity of all the input minterms. Doing so will hide the stripped minterms' sensitivities among other minterms and hence their sensitivities will no longer be `outliers' and reduce the sensitivity-based attack to a brute-force attack.
As SFLL-flex remains relatively secure, provides higher flexibility in selecting protected cubes, and is more relevant to SAS, {\it we focus on SFLL-flex} in this paper.}

An SFLL-flex configuration can be described using the number of protected cubes, $c$, and the number of specified bits of each cube, $k$, denoted as SFLL-flex$^{c\times k}$.
The authors of~\cite{yasin2017provably} derived the following characteristics of a circuit locked with SFLL-flex$^{c\times k}$: (1) the fraction of input minterms whose output will be corrupted by a wrong key (\ie the ``error rate'' of a wrong key) is $c\cdot 2^{-k}$; and (2) the probability that a SAT attack finds the correct key within $q$ iterations is $q\cdot 2^{\lceil log_2 {c} \rceil-k}$.
We illustrate this relationship in Fig.~\ref{fig:SFLL_tradeoff}. 
As a higher SAT success probability indicates weaker SAT resilience, SFLL inherently suffers from a trade-off between SAT resilience and effectiveness.
\begin{figure}[h]
    \centering
    \includegraphics[width=0.3\textwidth, trim={0 0 18cm 0cm},clip]{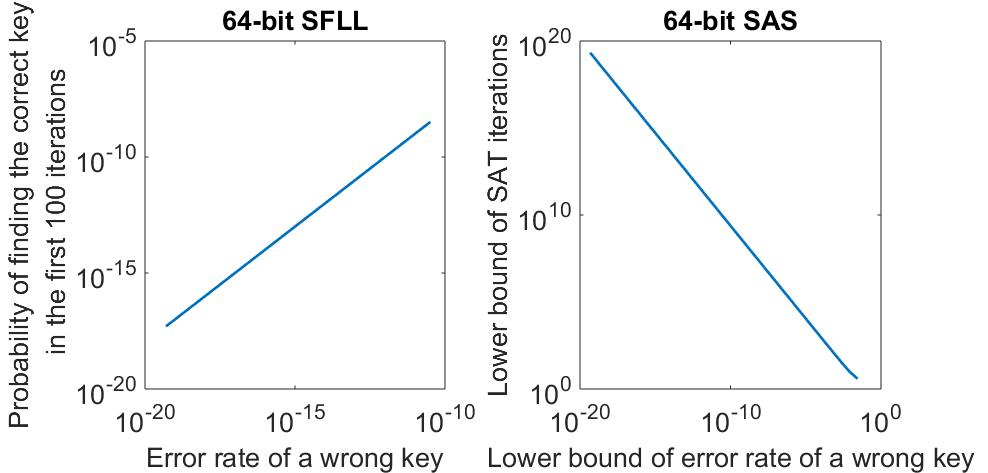}
    \caption{The positive correlation between the error rate of wrong keys and the probability that SAT finds the correct key in certain iterations}
    \label{fig:SFLL_tradeoff}
\end{figure}

\section{Insufficiency of SFLL for Real-World Applications}  \label{sec:attacks}

In this section, we investigate the application-level effectiveness of SFLL~\cite{yasin2017provably}.
Specifically, we lock the multiplier within a 32-bit 80386 processor since it is the largest combinational component.
The application-level effects are evaluated using both generic and machine learning (ML) benchmarks.
{\it We emphasize ML-based applications because they are inherently error-resilient and hence are more difficult to protect using logic locking.} 
Details of the benchmarks are listed in Table~\ref{tab:dataset_network}.

\begin{table}[h]
\centering
 \scriptsize
  \caption{Application benchmark details}
\begin{tabular}{|c|c|c|}
\hline
{\bf Benchmark Type} & {\bf Quantity} & {\bf Content} \\
\hline
Generic Applications & 9 & The PARSEC Benchmark Suite~\cite{bienia2008parsec} \\ \hline
Machine Learning & 5 & MNIST~\cite{lecun1998gradient}, SVHN~\cite{netzer2011reading}, CIFAR10~\cite{krizhevsky2009learning}, ILSVRC-2012~\cite{deng2012ilsvrc}, Oxford102~\cite{Nilsback08} \\ \hline
\end{tabular}%
  \label{tab:dataset_network}%
\end{table}%

{In order to evaluate the application-level impact of a logic locking scheme, we modify the GEM5~\cite{binkert2011gem5} simulator to reflect the effects of logic locking. Specifically, before simulation, based on the key value, the modified GEM5 simulator calculates the set of input minterms whose output will be corrupted according to the logic locking configuration. During the simulation, when the locked module gets any of these input minterms, the simulator randomizes the locked module's output to emulate locking-induced corruption. Under all other circumstances, the simulator works in the same way as GEM5.}
In this way, the circuit-level error induced by an incorrect (including approximate) key can be evaluated at the application level. 
This framework is illustrated in Fig.~\ref{fig:error_transfer} which is similar to the strategy used in~\cite{chakraborty2018gpu, zuzak2019memory}.

\begin{figure}[h]
    \centering
    \includegraphics[width=0.75\textwidth,trim={1 1 1 1}, clip]{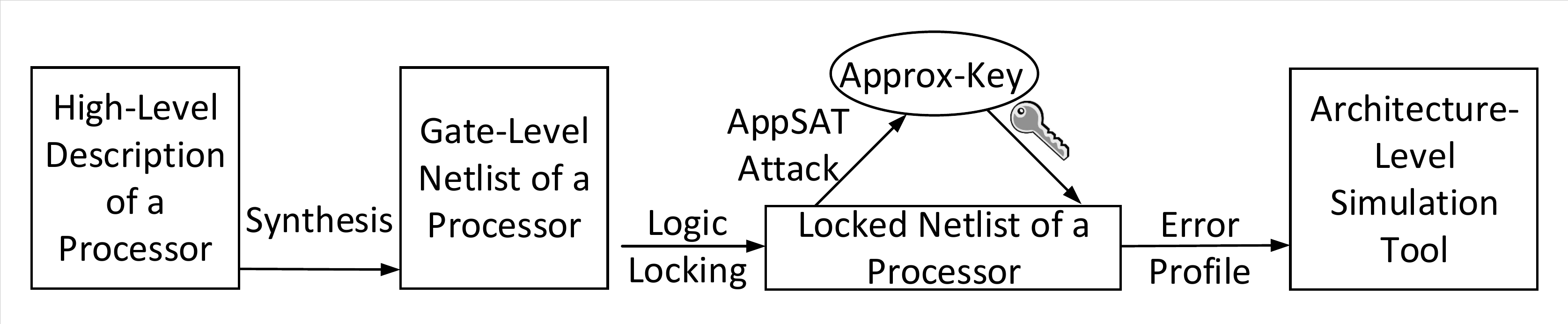}
    \caption{Our Experimental Framework}
    \label{fig:error_transfer}
\end{figure}

SFLL allows the designer to explore the trade-off between effectiveness and SAT resilience. We show that a ``sweet spot'' does not exist. In our experiments, we lock the multiplier with various SFLL configurations, each having a different level of SAT resilience, quantified by the average SAT iterations to unlock (as the X axis in Fig.~\ref{fig:SFLL_app_tradeoff}). 
The effectiveness of each locking scheme is evaluated by running the PARSEC and ML benchmarks on the locked processors loaded with approximate keys. The percentage of PARSEC benchmark runs with incorrect outcome and the accuracy loss of ML models are used as the criteria to evaluate the effectiveness of each locking configuration. The trade-off is illustrated in Fig.~\ref{fig:SFLL_app_tradeoff}.

\begin{figure}[htb]
    \centering
    \includegraphics[width=0.49\textwidth]{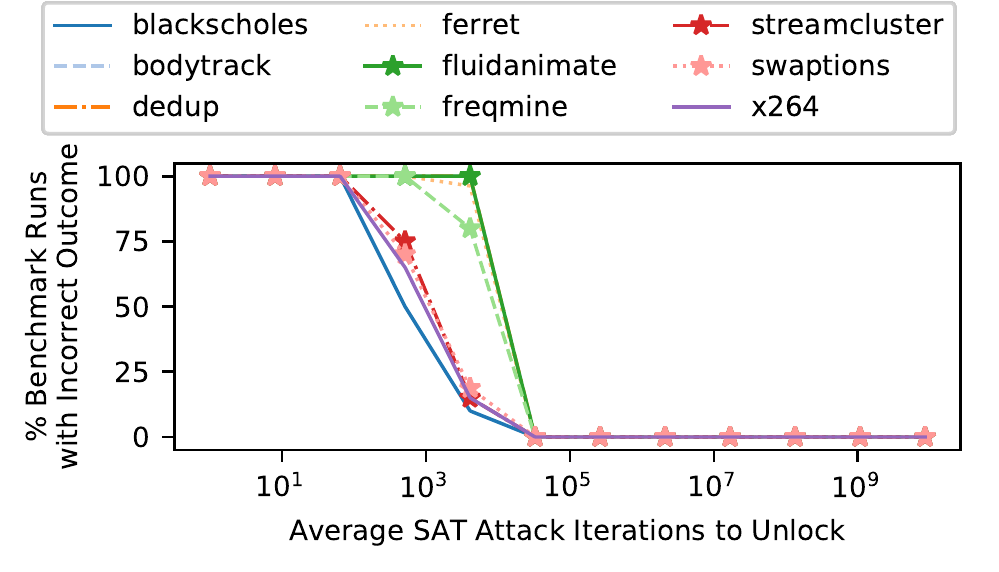}
    \includegraphics[width=0.49\textwidth]{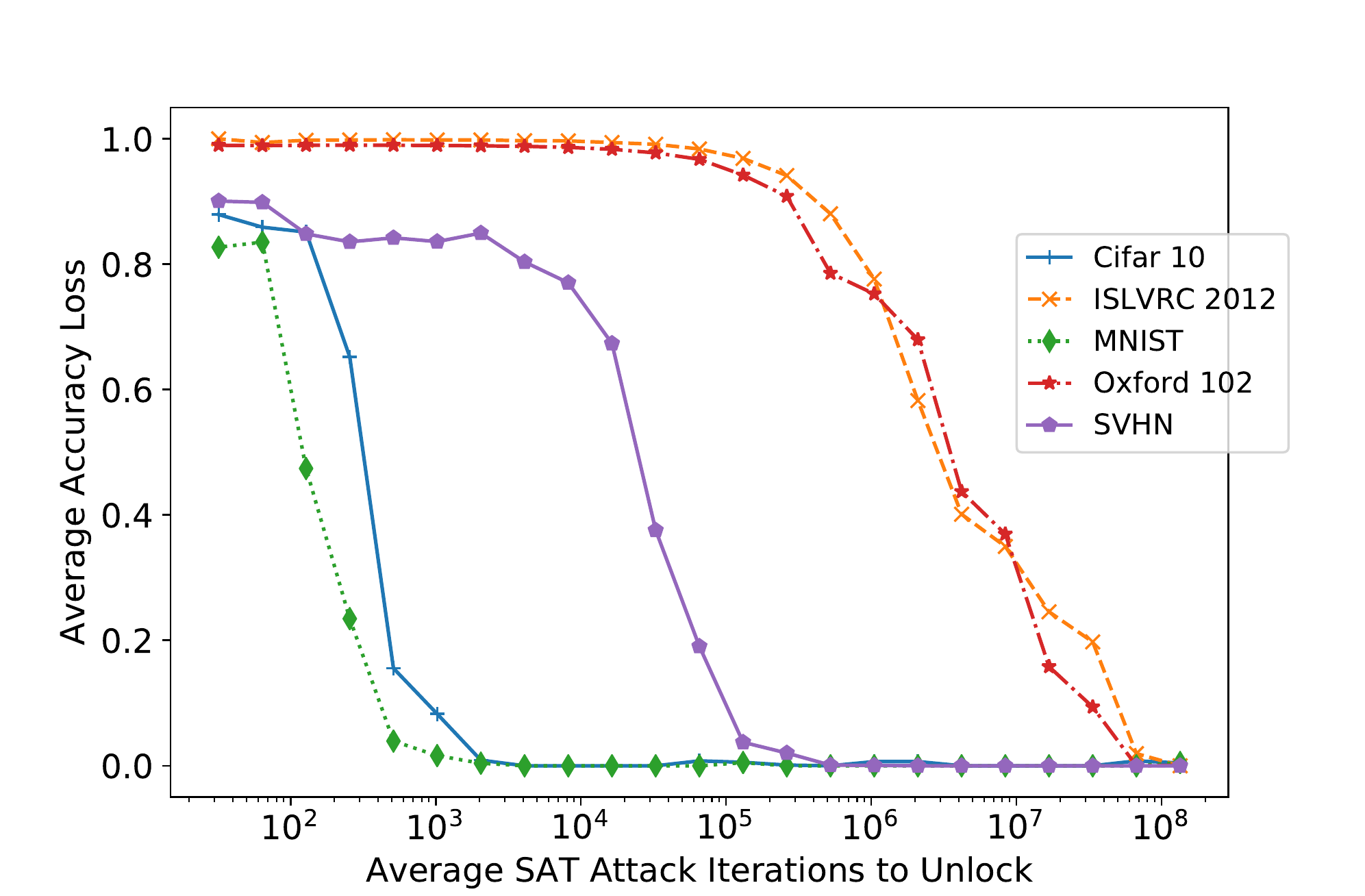}
    \caption{SAT resiliency vs. locking effectiveness trade-off. Left: PARSEC benchmarks. Right: ML benchmarks.}
    \label{fig:SFLL_app_tradeoff}
\end{figure}

From Fig.~\ref{fig:SFLL_app_tradeoff}, we observe that the wrong keys' impact decreases with the increase in SAT resiliency. In order to have a visible accuracy drop for the most error-resilient benchmarks, the SFLL locked processor cannot endure more than roughly 1000 SAT iterations. Such a locking scheme is extremely vulnerable since 1000 SAT iterations can be fulfilled within minutes. Therefore, a logic locking scheme that ensures high application-level impact without sacrificing SAT resiliency is needed.

\section{Fundamental Trade-off for All Logic Locking Schemes}\label{sec:trade-off}

This section generalizes the trade-off of SFLL to all logic locking schemes. We start with definitions of concepts and then prove the relationship between SAT resilience and effectiveness. {Recall that we use $F(\vec{X})$ to denote the Boolean functionality of the original circuit and $F_L(\vec{X},\vec{K})$ to denote the Boolean functionality of the locked circuit, where $\vec{X}$ is the input minterm and $\vec{K}$ is the key.

\begin{definition}
We say that a key $\vec{K}$ {\bf corrupts} an input minterm $\vec{X}$ if and only if the locked circuit's output is different from the original circuit's output when $\vec{X}$ is the input, \ie $F_L(\vec{X},\vec{K}) \ne F(\vec{X})$.
\end{definition}

\begin{definition}
The {\bf key error rate (KER)} $e_{\vec{K}}$ of a key $\vec{K}$ is defined as the number of input minterms corrupted by the key $\vec{K}$ divided by the total number of input minterms.
\end{definition}

Let $\XXX_{\vec{K}}$ be the set of input minterms corrupted by $\vec{K}$. Then,
\begin{equation}
    e_{\vec{K}} = \frac{|\XXX_{\vec{K}}|}{2^n}
\end{equation}
where $n$ is the number of bits in the input. 

We use $e_w$ to denote the average KER across all the {\it wrong keys}.
\begin{equation}
    e_w = \frac{1}{|\KKK^W|}\sum_{\forall\vec{K}\in\KKK^W}e_{\vec{K}}
\end{equation}

\begin{definition}
The {\bf input error rate (IER)} $\gamma_{\vec{X}}$ of an input minterm $\vec{X}$ is the number of wrong keys that corrupt this minterm divided by the total number of wrong keys.
\end{definition}
}
Let $\KKK_{\vec{X}}$ be the set of wrong keys that corrupt the input minterm $\vec{X}$ and $\KKK^W$ be the set of all wrong keys. Then, 
\begin{equation}
    \gamma_{\vec{X}} = \frac{|\KKK_{\vec{X}}|}{|\KKK^W|}
\label{eq:ier_def}
\end{equation}
Let $\gamma$ denote the average IER over all the input minterms, \ie
\begin{equation}
    \gamma = \frac{1}{2^n}\sum_{\forall \vec{X}\in\{0,1\}^n}\gamma_{\vec{X}}
\label{eq:avg_ier_def}
\end{equation}

Let us illustrate the above concepts with the following example. We consider a circuit with two input bits $(x_0, x_1)$ and locked with a two-bit key $(k_0, k_1)$, as shown in Fig.~\ref{fig:locking_example}. Table~\ref{tab:locking_truth_table} is the truth table for each possible input minterm and key. If a key corrupts an input minterm, the corresponding cell is marked with (\xmark).
\begin{figure}[h]
    \centering
    \includegraphics[width=0.6\textwidth, trim={1 1 1 1}, clip]{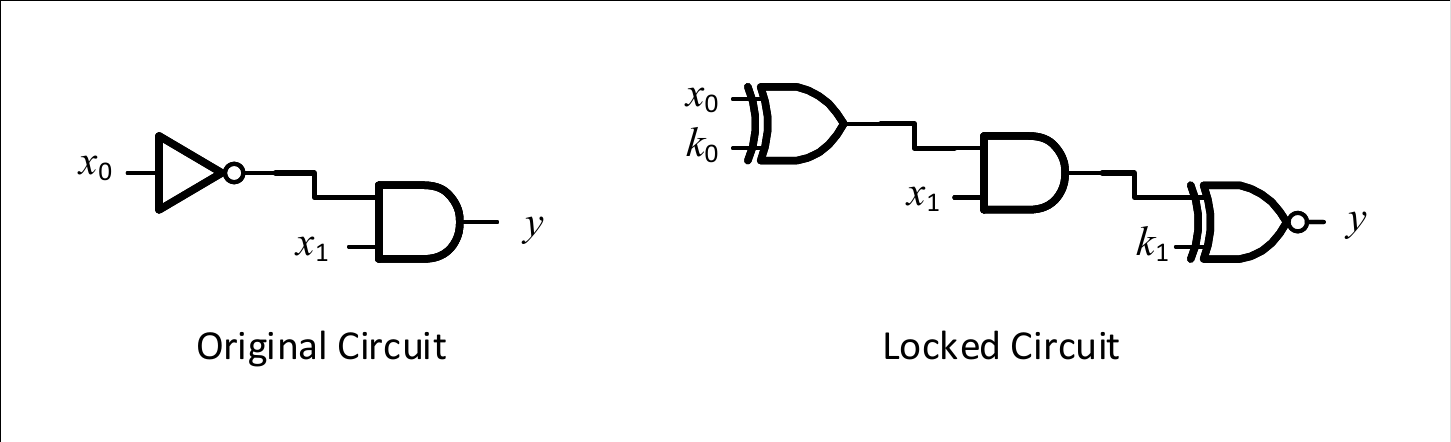}
    \caption{An example of logic locking, with the original circuit on the left and the locked circuit on the right.}
    \label{fig:locking_example}
\end{figure}
\begin{table}[h]
    \centering
    \scriptsize
    \caption{Truth table of the locked circuit in Fig.~\ref{fig:locking_example}}
    \begin{tabular}{|c|c|c|c|c|c|c|c|}
    \hline
     & $\vec{K}=(0,0)$ & $\vec{K}=(0,1)$ & \cellcolor{green!25}$\vec{K}=(1,1)$ & $\vec{K}=(1,0)$ & Correct $y$ & $\gamma_{\vec{X}}$ & $\gamma$\\ \hline
     $\vec{X}=(0,0)$ & 1(\xmark) & 0 & \cellcolor{green!25}0 & 1(\xmark) & 0 & $\frac{2}{3}$ & \multirow{4}{*}{$\frac{2}{3}$} \\    
    \cline{1-7}
     $\vec{X}=(0,1)$ & 1 & 0(\xmark) & \cellcolor{green!25}1 & 0(\xmark) & 1 & $\frac{2}{3}$ & \\
     \cline{1-7}
     $\vec{X}=(1,1)$ & 0 & 1(\xmark) & \cellcolor{green!25}0 & 1(\xmark) & 0 & $\frac{2}{3}$ & \\     
     \cline{1-7}
     $\vec{X}=(1,0)$ & 1(\xmark) & 0 & \cellcolor{green!25}0 & 1(\xmark) & 0 & $\frac{2}{3}$ & \\     
     &&&&&\multicolumn{3}{c}{} \\ [-2.43ex]
     \hline
    &&&&&\multicolumn{3}{c}{} \\ [-2.43ex]
     $e_{\vec{K}}$ & $\frac{1}{2}$ & $\frac{1}{2}$ & \cellcolor{green!25} 0 & 1 & \multicolumn{3}{c}{} \\     &\multicolumn{4}{c|}{}&\multicolumn{3}{c}{} \\ [-2.43ex]
     \cline{1-5}
    &\multicolumn{4}{c|}{}&\multicolumn{3}{c}{} \\ [-2.43ex]
     $e_w$ & \multicolumn{4}{c|}{$\frac{2}{3}$}& \multicolumn{3}{c}{} \\ 
     \multicolumn{8}{c}{}\\ [-2.43ex] \cline{1-5}
    \end{tabular}
    \label{tab:locking_truth_table}
\end{table}

In Table~\ref{tab:locking_truth_table}, we also calculate the KER of each key, the IER of each input minterm, and their averages. We can observe that both the average KER and average IER equal $\frac{2}{3}$. It turns out that this equality is universal in logic locking:

\begin{theorem}
The average KER of all wrong keys equals the average IER of all input minterms, \textit{i. e.} $e_w = \gamma$.
\label{thm:equivalence}
\end{theorem}

The proof is in Appendix~\ref{proof:thm_equi}.

Let $\lambda$ be the number of SAT iterations that are needed to find the correct key.
\begin{theorem}
The expected number of SAT iterations $E[\lambda]$ is lower bounded by $\frac{1}{\gamma}$.
\label{thm:iterations}
\end{theorem}
\begin{proof}
{Recall that $\gamma$ is the average IER of all input minterms, defined in Equation \eqref{eq:avg_ier_def}. In each SAT iteration, a distinguishing input (DI) is found and \textit{all} the wrong keys that corrupt this DI will be pruned out from the key search space. The average number of such wrong keys for each DI is $\gamma |\KKK^W|$. Because some of the wrong keys may have already been pruned out in previous iterations, the average number of wrong keys pruned out in each SAT iteration is at most $\gamma |\KKK^W|$. SAT attack finishes when every wrong key is pruned out, hence the average number of SAT iterations can be lower bounded by:}
\begin{equation}
    E[\lambda] \ge \frac{|\KKK^W|}{\gamma |\KKK^W|}=\frac{1}{\gamma}
\end{equation}
Hence proved.
\end{proof}

Theorems~\ref{thm:equivalence} and~\ref{thm:iterations} explicitly point out that there exists an inverse relationship between $e_w$ and the lower bound of E[$\lambda$]. This quantifies the trade-off between them. This trade-off applies to any logic locking scheme.
Note that different input minterms may inject a different amount of error at the application level. By assigning higher IER to a few minterms with high application-level impact, we can achieve high effectiveness while maintaining high SAT resilience by keeping $\gamma$ low and $E[\lambda]$ high. This is the main intuition behind SAS.

\section{The Architecture and Properties of SAS} \label{sec:countermeasure}

{Theorems \ref{thm:equivalence} and \ref{thm:iterations} have expressed a relationship between input \& key error rates and the number of SAT iterations. These quantities are related to two objectives of logic locking:}

\begin{enumerate}
\item {\bf Effectiveness:} Any incorrect key should have a high application-level error impact.
\item {\bf SAT resilience:} The complexity of determining the correct key via SAT attacks should be very high.
\end{enumerate}

{Assuming that each SAT iteration takes constant time, these two objectives may compete with each other.}
In this section, we introduce \emph{Strong Anti-SAT (SAS)} logic locking scheme which aims to achieve both objectives simultaneously. {SAS guarantees that the expected complexity of SAT attack increases exponentially in the size of the locking key} while having a significant impact on the accuracy of real-world applications. In SAS, instead of uniformly distributing the error across all possible inputs, we identify certain input patterns which potentially have a higher impact on the overall application-level error. We call these inputs {\bf critical minterms}. SAS is configured in such a way that any incorrect key corrupts at least 1 critical minterm, resulting in high input error rate (IER) for critical minterms. For the other minterms, the IER is low.

\subsection{The SAS Block}
Let $\MMM$ be the set of critical minterms and $m=|\MMM|$ be the number of critical minterms.
For the ease of implementation, we always choose $m$ to be a power of 2.
The basic locking infrastructure is the \emph{SAS} block which is illustrated in Fig.~\ref{SELL_1}.
{When the input minterm $\vec{X}$ has $n$ bits, the key $\vec{K}$ has $2n$ bits. Let $\vec{K}_1$ and $\vec{K}_2$ be the first $n$ bits and the second $n$ bits of $\vec{K}$, respectively, since they play different roles in SAS.}
In order to describe the mechanism of the SAS locking scheme clearly, we use a reverse order and start our illustration from the output side.

\begin{figure}[h]
\centering
\includegraphics[width=0.5\textwidth, trim={1 1.5cm 1 1.4cm}, clip]{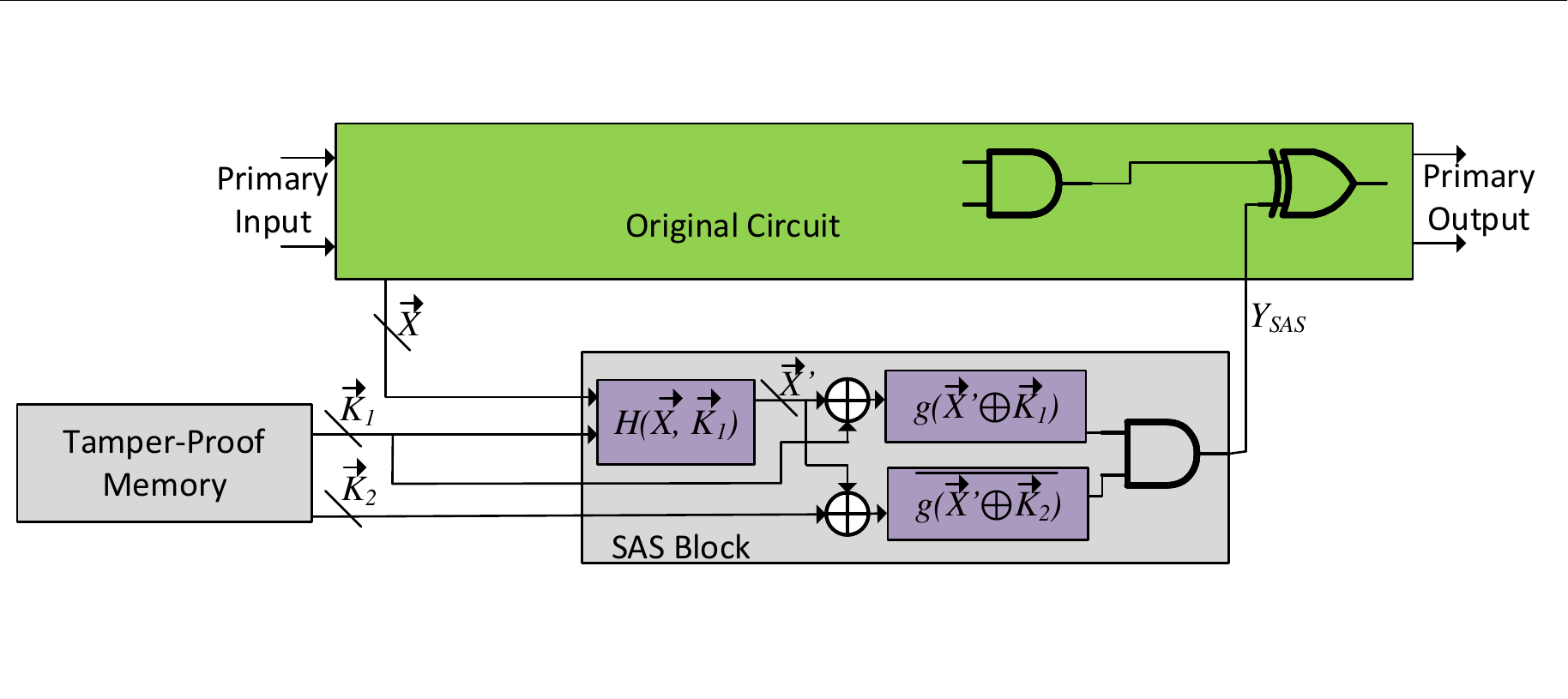}
\caption{The Architecture of SAS Configuration 1 with the Details of the SAS Block}
\label{SELL_1}
\end{figure}

$Y_{SAS}$ is the output of the SAS block. If $Y_{SAS}=1$, a fault will be injected into the original circuit. $Y_{SAS}$ is the output of an AND gate, whose inputs are two function blocks with opposite functionalities, namely $g$ and $\bar g$.
$g$ is a function with an on-set-size of 1, \ie only one input minterm will have output 1 and all others will have output 0. {We call this very input value $\vec{X}_g$, which can be \textit{any} $n$-bit value determined by the designer. $\bar{g}$ has the opposite functionality of $g$, \ie only when $\vec{X}=\vec{X_g}$ will $\bar g(\vec{X})=0$. Due to the AND gate which produces $Y_{SAS}$, in order to corrupt an input minterm, the outputs of both $g$ and $\bar{g}$ need to be 1.

Each input of $g$ and $\bar g$ is the output of a XOR or XNOR gate and the designer can choose either one for each gate. We use ``$\oplus$'' to denote this bit-wise operation. It can be observed that the output AND gate, the $g$ and $\bar g$ functions, and the $\oplus$ gates constitute an Anti-SAT block. Notice that, the $\oplus$ operations before $g$ and $\bar g$ need not be the same. However, without losing generality, we assume that they are the same for the ease of the following discussions.

A function block $\vec{X}'=H(\vec{X},\vec{K}_1)$ is inserted before $g$ and $\bar{g}$. This block will determine the IER of $\vec{X}$ and it works as follows. 
If $\vec{X}$ is not a critical minterm, \ie $\vec{X}\notin\MMM$, then $\vec{X}'=\vec{X}$ \ie the input minterm simply passes through $H$. In this case, a wrong key must satisfy two conditions to corrupt input minterm $\vec{X}$: $\vec{K}_1=\vec{X}\oplus\vec{X}_g$ and $\vec{K}_2\ne\vec{X}\oplus\vec{X}_g$. 
Thus there are $2^n-1$ wrong keys (1 possible $\vec K_1$ combined with $2^n-1$ possible $\vec{K}_2$) that corrupt $\vec{X}$, out of the total number of $2^n(2^n-1)$ wrong keys. 
Therefore, the IER of a non-critical input minterm $\vec{X}$, $\gamma_{\vec{X}}$, is very low: $\gamma_{\vec{X}}=\frac{2^n-1}{2^n(2^n-1)}=2^{-n}$.

If $\vec{X}$ is a critical minterm, \ie $\vec{X}\in\MMM$, its IER will be $\gamma_{\vec{X}}=\gamma_c=\frac{1}{m}$, where $\gamma_c$ denotes the IER of each critical minterm and $m$ is the total number of critical minterms. This is achieved as follows. For $\frac{2^n}{m}$ values of $\vec{K}_1$, $H(\vec{X},\vec{K}_1)=\vec X_g$. We use $\KKK^1_{\vec{X}}$ to denote this set of $\vec{K}_1$ values. 
For other $\vec{K}_1$ values, the input minterm will still pass through $H$, \ie $H(\vec{X},\vec{K}_1)=\vec X$.
Moreover, the $\KKK^1_{\vec{X}}$ sets for each critical minterm are mutually exclusive and evenly partition $\{0,1\}^n$, \ie 
\begin{equation}
    \forall \vec{X_1},\vec{X_2} \in \MMM, \
    |\KKK^1_{\vec{X_1}}| = |\KKK^1_{\vec{X_2}}|, \
    \KKK^1_{\vec{X_1}}\land \KKK^1_{\vec{X_2}}=\emptyset \text{, and }
    \bigcup_{\forall \vec{X}\in\MMM} \KKK^1_{\vec{X}}=\{0,1\}^n
\label{eq:wrong_key_partition}
\end{equation}
where $n$ is the number of bits in $\vec{X}$, $\vec{K}_1$, and $\vec{K}_2$. Notice that, there are many possibles implementations of the partition of $\vec{K}_1$ space that satisfies Equation \eqref{eq:wrong_key_partition}. For example, if there are two critical minterms $\vec{X}_1$ and $\vec{X}_2$, the designer can let $\KKK^1_{\vec{X}_1}$ be the set of $\vec{K}_1$ values whose most significant bit (MSB) is 1 and $\KKK^1_{\vec{X}_2}$ be the set of $\vec{K}_1$ values whose MSB is 0.
In Table~\ref{case1table}, we demonstrate how the space of $\vec{K}_1$ is partitioned to corrupt each critical and non-critical input minterm. In the first row of Table~\ref{case1table}, we use the indices to indicate \textit{how many} $K_1$ values there are that can lead a wrong key to corrupt each input minterm. The indices can accord to any ordering of the elements in $\{0,1\}^n$.
}

\begin{table}[h]
\centering
\scriptsize
\caption{Illustration of how $m$ critical minterms partition the set of wrong keys}
\begin{tabular}{p{0.8cm} p{0.5cm}| >{\centering}p{0.3cm} >{\centering}p{0.3cm} >{\centering}p{0.3cm}| >{\centering}p{0.4cm} >{\centering}p{0.4cm} >{\centering}p{0.4cm}| >{\centering}p{0.2cm} p{0.2cm}}
\multicolumn{2}{c|}{$\vec{K}_1$ of wrong keys}& $\vec{k}_1$ & $\cdots$ & $\vec{k}_{\frac{2^n}{m}}$ & $\vec{k}_{\frac{2^n}{m}+1}$ & $\cdots$ & $\vec{k}_{2\frac{2^n}{m}}$ & $\cdots$ &$\vec{k}_{2^n}$  \\ \hline
\multirow{4}{*}{\parbox{1cm}{critical minterms}}
&$\vec{X}_1$ & $\bullet$ & $\bullet$ & $\bullet$ & &&&& \\
&$\vec{X}_2$ & & & &$\bullet$&$\bullet$&$\bullet$& & \\
&$\cdots$ & & & & & & $\cdots$ & &  \\
&$\vec{X}_m$ & & & & & & &$\bullet$&$\bullet$ \\\hline
\multirow{4}{*}{\parbox{1cm}{non-critical minterms} }
&$\vec{X}_{m+1}$ & $\bullet$ &&&&  &&& \\
&$\vec{X}_{m+2}$ && $\bullet$ &&&&  && \\
&$\cdots$ &&&& & $\ddots$ & & & \\
&$\vec{X}_{2^n}$ &&&&&&&& $\bullet$ \\
\end{tabular}
\label{case1table}
\end{table}

The 2 configurations of SAS will be introduced in the rest of this section.

\subsection{Configuration 1: SAS with One SAS Block}

This configuration is illustrated in Fig.~\ref{SELL_1}.
In this configuration, there is one SAS block.
As the critical minterms evenly partition the set of wrong keys, the IER of each critical minterm is $\gamma_c=\frac{1}{m}$.
Below we derive the SAT resilience of this configuration assuming that the SAT solver chooses a DI uniformly at random in each iteration.
This is a common assumption~\cite{yasin2017ttlock, yasin2017provably, sengupta2020truly}.
The SAT resilience is quantified using the expected number of SAT iterations $E[\lambda]$.
To start with, we give 2 useful lemmas.

\begin{lemma}
Let $\DDD^i$ be the set of DIs that have been chosen in the first $i$ iterations and $\vec{X}$ be a primary input minterm.
If all the wrong keys that corrupt $\vec{X}$ have already been pruned out in the previous SAT iterations, \ie $\KKK_{\vec{X}} \subset \bigcup_{\forall \vec{X'}\in \DDD^i}\KKK_{\vec{X'}}$, then $\vec{X}$ cannot be the DI of any SAT iteration beyond $i$.
\label{lemma:removal}
\end{lemma}
The proof is given in Appendix~\ref{proof:lemma_removal}.

\begin{lemma}
For SAS Configuration 1, any critical minterm must exist in the set of DIs when SAT finishes: $\vec{X} \in \DDD^\lambda \  \forall \vec{X}\in \MMM$, where $\lambda$ is the total number of SAT iterations and $\DDD^\lambda$ is the set of all DIs.
\label{lemma:critical_count}
\end{lemma}
The proof is given in Appendix~\ref{proof:lemma_critical_count}.

\begin{theorem}
The expected number of SAT iterations of SAS Configuration 1 is
\begin{equation}
E[\lambda] = \frac{2^n+m}{2}
\label{mean_case2}
\end{equation}
\end{theorem}
\begin{proof}
The total number of SAT iterations equals the total number of DIs. DIs consist of critical minterms and non-critical minterms. By Lemma~\ref{lemma:critical_count}, all the critical minterms must be in the set of DIs for SAT to terminate.
Therefore, we only need to find the expected number of \emph{non-critical minterms} that are chosen as DIs.
As illustrated in Table~\ref{case1table},
$\forall \vec{X'}\notin\MMM$, $\exists$ exactly one $\vec{X}\in\MMM$ such that $\KKK_{\vec{X'}}\subset\KKK_{\vec{X}}$.
By Lemma~\ref{lemma:removal}, if this $\vec{X}$ is chosen as DI before $\vec{X'}$, then $\vec{X'}$ cannot be chosen in further iterations any more.
In other words, $\vec{X'}$ will count towards the total number of iterations only when it is chosen before the critical minterm $\vec{X}$.
By our assumption that the DI is chosen uniformly at random in each iteration, $\vec{X'}$ has a probability of $\frac{1}{2}$ to be chosen as DI before $\vec{X}$ is chosen.
As this is true for any non-critical minterm, the expected number of SAT iterations is $E[\lambda]=\frac{1}{2}(2^n-m)+m=\frac{2^n+m}{2}$.
\end{proof}

\subsection{Configuration 2: SAS with Multiple Blocks}

In this configuration, we have $l$ SAS blocks as illustrated in Fig.~\ref{SELL_3}.
The $n$-bit primary input $\vec{X}$ is shared among all the SAS blocks and there is a $2n$-bit key input for each SAS block.
The output of each SAS block is XOR'ed with a wire in the original circuit.
Therefore, a fault is injected into the original circuit if any SAS block has output 1.
Let $\MMM^j$ be the set of critical minterms for the $j^\text{th}$ SAS block
, $j=1,2,\ldots,l$.
For ease of implementation, we choose $l$ also to be a power of 2 and $l\le m$.
The relationship between $\MMM^j$ and the total set of critical minterms $\MMM$ is that
\emph{$\MMM^1,\MMM^2,\ldots,\MMM^l$ have the same cardinality, are mutually exclusive, and evenly partition $\MMM$,} \ie
\begin{equation}
    |\MMM^1|=|\MMM^2|=\cdots=|\MMM^l|,\ \MMM^i \cap \MMM^j = \emptyset\ \forall i\ne j, \text{ and } \bigcup_{k=1}^{l}\MMM^k=\MMM
\end{equation}
In this way, each SAS block has $\frac{m}{l}$ critical minterms. As each critical minterm receives high IER from only one SAS block, the IER of any critical minterm is $\gamma_c=\frac{l}{m}$.
\begin{figure}[h]
\centering
\includegraphics[width=0.7\textwidth, trim={1 1 1 1}, clip]{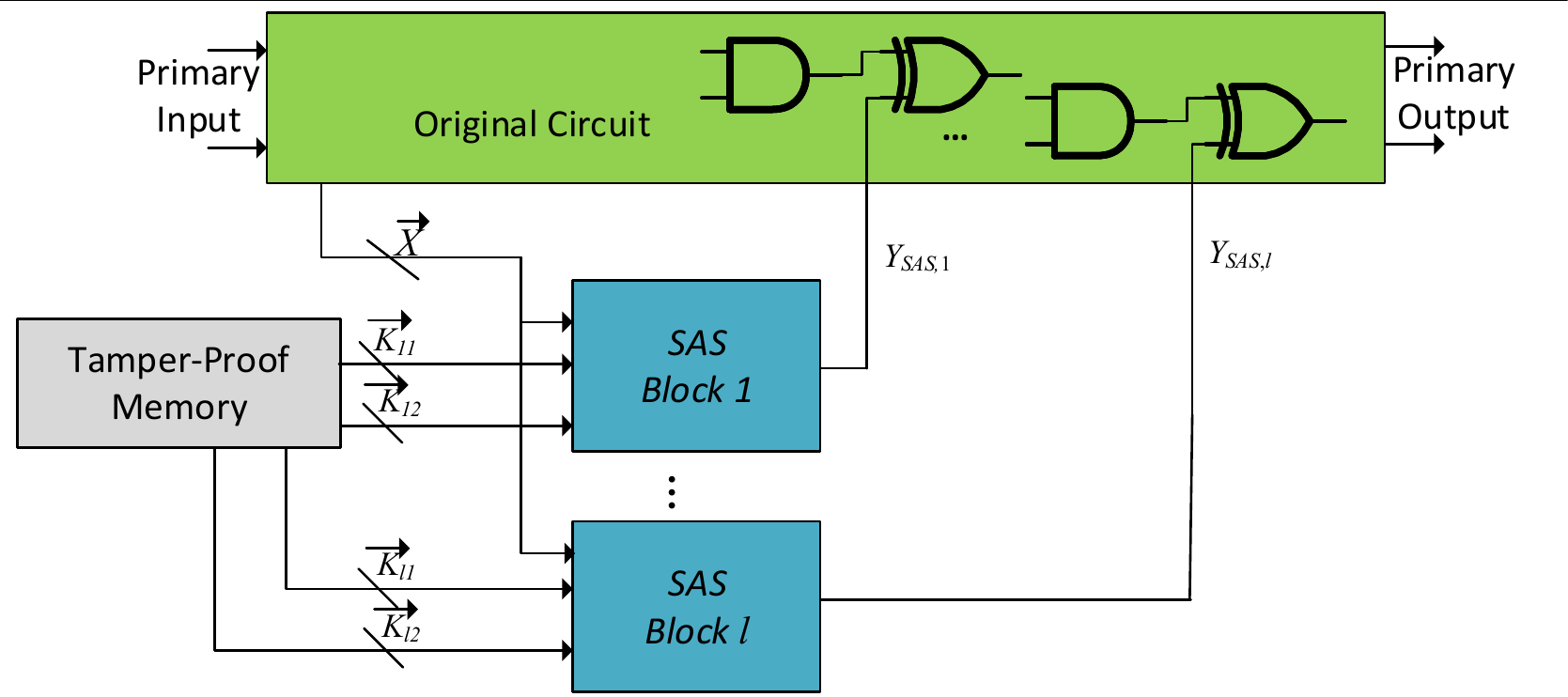}
\caption{Configuration 2 with $l$ SAS blocks}
\label{SELL_3}
\end{figure}

\begin{lemma}
For SAS Configuration 2, any critical minterm must exist in the set of DIs when SAT finishes: $\vec{X} \in \DDD^\lambda \  \forall \vec{X}\in \MMM$, where $\lambda$ is the total number of SAT iterations and $\DDD^\lambda$ is the set of all DIs.
\label{lemma:SAS_critical_count}
\end{lemma}
The proof is given in Appendix~\ref{proof:lemma_sas_crit}.
Below, we will analyze the SAT resilience of this configuration by deriving the expected number of SAT iterations.

\begin{theorem}
The expected number of SAT iterations of SAS Configuration 2 with $l$ SAS blocks and $m$ critical minterms is
\begin{equation}
    E[\lambda] = \frac{l\cdot 2^n+m}{l+1}
\label{mean_case4}
\end{equation}
\end{theorem}
\begin{proof}
By Lemma~\ref{lemma:SAS_critical_count}, every critical minterm must count toward the total number of SAT iterations.
Therefore, we only need to derive the expected number of non-critical minterms that are chosen as DIs.

For any non-critical minterm $\vec{X'}\notin\MMM$, in the $i^\text{th}$ SAS  block, there exists exactly one critical minterm $\vec{X}_i$ such that the set of wrong keys that corrupt $\vec{X'}$ in this SAS block, $\KKK_{i,\vec{X'}}$, is a subset of the set of wrong keys that corrupt $\vec{X}_i$, $\KKK_{i,\vec{X}_i}$. \ie $\KKK_{i,\vec{X'}} \subset \KKK_{i,\vec{X}_i}$.
As the construction of the SAS block makes this true for any individual SAS block and the critical minterms for each SAS block are mutually exclusive, there are a total of $l$ such critical minterms.
When {\it all of these $l$ critical minterms} are chosen as DI, they will cover the entire set of wrong keys that corrupt $\vec{X'}$. Therefore, by Lemma~\ref{lemma:removal}, in order to include $\vec{X'}$ in the set of DIs, it must be selected before all $l$ critical minterms are selected. This holds for any non-critical minterm.

By our assumption that the DIs are chosen uniformly at random in each SAT iteration, the probability that each non-critical minterm will be chosen as DI is $\frac{l}{l+1}$.
Therefore, the expected number of SAT iterations is $E[\lambda]=\frac{l}{l+1}(2^n-m)+m=\frac{l\cdot 2^n+m}{l+1}$.
\end{proof}

The properties of both configurations of SAS are summarized in Table~\ref{tab:configs}. 

\begin{table}[h]
\centering
\scriptsize
\caption{Properties of the 2 Configurations of SAS}
\begin{tabular}{|c| c| c| c|}
\hline
&&&\\ [-2ex]
Configuration & $l$ & $\gamma_c$ & $E[\lambda]$ \\[0.5ex] \hline
&&&\\ [-2ex]
1 & 1 & $\frac{1}{m}$ & $\frac{2^n+m}{2}$ \\ [1ex]  \hline
&&&\\ [-2ex]
2 & $1\leq l \leq m$ & $\frac{l}{m}$ & $\frac{l2^n+m}{l+1}$ \\ [1ex]  \hline
\end{tabular}
\label{tab:configs}
\end{table}

\section{Robust SAS: a Removal-Resilient SAS Variant}
\label{sec:RSAS}
Although SAS achieves desirable SAT resilience and high IER on critical minterms, it is still vulnerable to removal attack. In such an attack, the attacker can identify and remove each SAS block and replace their output wires with constant 0. In this way, the remaining part of the locked circuit will have correct functionality.
In order to address this drawback, we introduce Robust SAS (RSAS), a variant of SAS that is resilient to removal attacks. In addition to adding an RSAS function block, RSAS modifies the functionality of the original circuit. Therefore, unlike SAS, one cannot obtain the correct functionality of the circuit by identifying and removing the RSAS block.
We will introduce the architecture of RSAS and show how any SAS configuration can be converted to a functionally equivalent RSAS configuration.
Due to the equivalence in functionality, an RSAS configuration will have the same \textbf{SAT resilience} and \textbf{effectiveness} as its SAS counterpart.

\subsection{RSAS Architecture and Relationship with SAS}
\label{ssec:sas_to_rsas}
A circuit locked by RSAS consists of an altered original circuit and one or more RSAS block(s). Fig. \ref{fig:RSAS_1} illustrates the RSAS configuration with one RSAS block. Given the same set of critical minterms and the same number of locking function blocks, locking a circuit with RSAS and SAS will yield the same functionality.
An RSAS-locked circuit can be obtained by converting a functionally equivalent SAS-locked circuit in the following way.

\begin{figure}[htb]
    \centering
    \includegraphics[width=.7\textwidth, trim={1 1 1 1}, clip]{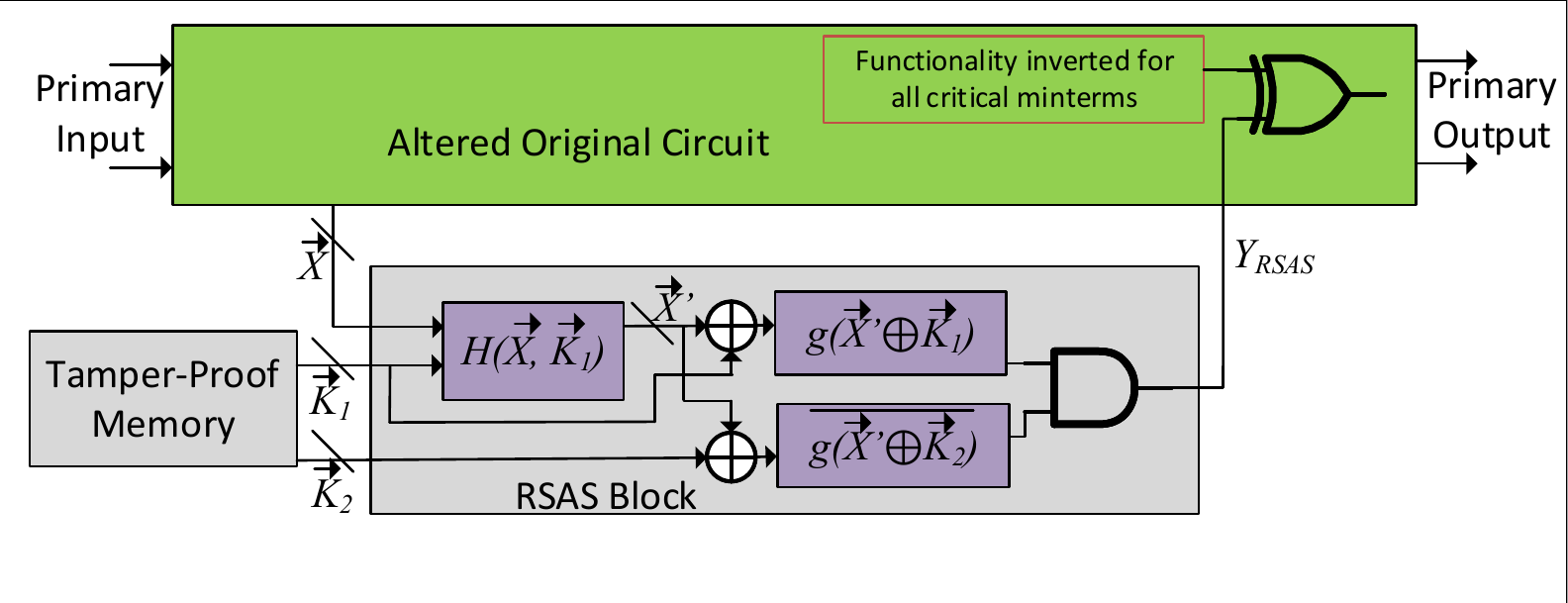}
    \caption{A circuit locked with one RSAS block, equivalent to SAS Configuration 1}
    \label{fig:RSAS_1}
\end{figure}

\subsubsection{Altering the original circuit}
Recall that $l$ is the number of SAS blocks in a SAS configuration. For the $j^\text{th}$ SAS block, $j=1,2,\ldots,l$, the set of critical minterms it contains is denoted by $\MMM^j$ and $|\MMM^j|=\frac{m}{l}$, where $m$ is the total number of critical minterms. 
In order to implement RSAS, we need to modify the original circuit's functionality. Notice that, for each SAS block, there is a wire in the original circuit that is XOR'ed with the SAS block's output. 
For the $j^\text{th}$ SAS block, we locate this wire. For each critical minterm in $\MMM^j$, we invert the functionality of critical minterms at this wire. This needs to be done for each $j$ in $j=1,2,\ldots,l$.
This is illustrated in Fig. \ref{fig:rsas_multi}.

\begin{figure}[htb]
    \centering
    \includegraphics[width=0.7\textwidth, trim={1 1 1 1}, clip]{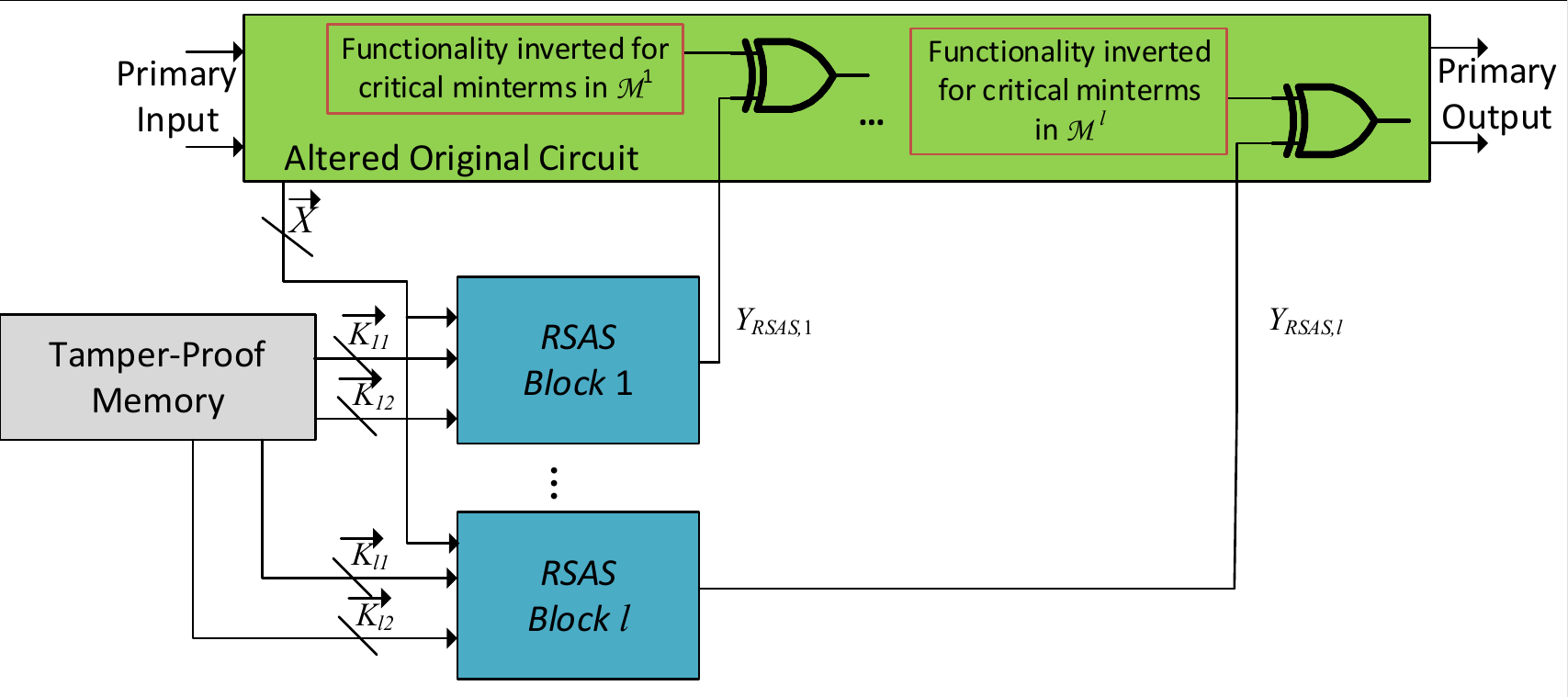}
    \caption{A circuit locked with multiple RSAS blocks, equivalent to SAS Configuration 2}
    \label{fig:rsas_multi}
\end{figure}

\subsubsection{Converting the SAS block into the RSAS block}
The RSAS block is very similar to the SAS block and there is only one difference between them. 
For the $j^\text{th}$ SAS block, $j=1,2,\ldots,l$, if the primary input is a critical minterm in $\MMM^j$, the output of RSAS block, $Y_{RSAS,j}$, is the inversion of the output of SAS block, $Y_{SAS,j}$. 
Recall that, for a SAS configuration with $m$ critical minterms and $l$ SAS blocks, each critical minterm's IER is $\gamma_c=\frac{l}{m}$. Hence for a portion of $\frac{l}{m}$ wrong keys, $Y_{SAS,j}$ is 1.
This is achieved by the $\vec{X}'=H(\vec{X},\vec{K}_1)$ function: if $\vec{X}$ is a critical minterm, then the $H(\vec{X},\vec{K}_1)$ function makes sure that for $\gamma_c$ portion of wrong keys, we will have $g(\vec{X}'\oplus\vec{K}_1)=1$.
For RSAS, since the functionality for critical minterms is inverted, the portion of wrong keys that makes $Y_{RSAS,j}$ be 1 is $1-\gamma_c=\frac{m-l}{m}$.
This means the functionality of $H(\vec{X},\vec{K}_1)$ needs to be modified in the following way: if $\vec{X}$ is a critical minterm, then for $1-\gamma_c$ portion of wrong keys, $g(\vec{X}'\oplus\vec{K}_1)$ will output 1.
For non-critical input minterms, $Y_{RSAS}$ behaves in the same as $Y_{SAS}$. 
This is illustrated in Table \ref{tab:rsas}.

\begin{table}[htb]
\centering
\scriptsize
\caption{Illustration of RSAS block's functionality. A `$\bullet$' stands for $Y_{RSAS}=1$.}
\begin{tabular}{p{0.8cm} p{0.4cm}| >{\centering}p{0.3cm} >{\centering}p{0.3cm} >{\centering}p{0.3cm}| >{\centering}p{0.35cm} >{\centering}p{0.35cm} >{\centering}p{0.35cm}| >{\centering}p{0.2cm} p{0.2cm}}
\multicolumn{2}{c|}{$\vec{K}_1$ of wrong keys}& $\vec{k}_1$ & $\cdots$ & $\vec{k}_{\frac{2^n}{m}}$ & $\vec{k}_{\frac{2^n}{m}+1}$ & $\cdots$ & $\vec{k}_{2\frac{2^n}{m}}$ & $\cdots$ &$\vec{k}_{2^n}$  \\ \hline
\multirow{4}{*}{\parbox{1cm}{critical minterms}}
&$\vec{X}_1$ &  &  &  & $\bullet$ & $\bullet$ & $\bullet$ & $\bullet$ & $\bullet$ \\
&$\vec{X}_2$ & $\bullet$ & $\bullet$ & $\bullet$ & & & & $\bullet$ & $\bullet$\\
&$\cdots$ &$\bullet$ & $\bullet$ & $\bullet$&$\bullet$ & $\bullet$ & $\bullet$ & $\bullet$ & $\bullet$  \\
&$\vec{X}_m$ & $\bullet$ & $\bullet$ & $\bullet$ & $\bullet$ & $\bullet$ & $\bullet$ & & \\\hline
\multirow{4}{*}{\parbox{1cm}{non-critical minterms} }
&$\vec{X}_{m+1}$ & $\bullet$ &&&&  &&& \\
&$\vec{X}_{m+2}$ && $\bullet$ &&&&  && \\
&$\cdots$ &&&& & $\ddots$ & & & \\
&$\vec{X}_{2^n}$ &&&&&&&& $\bullet$ \\
\end{tabular}
\label{tab:rsas}
\end{table}

\subsection{SAT Resilience and Effectiveness of RSAS}
In Sec. \ref{ssec:sas_to_rsas}, we introduced how to convert a SAS-locked circuit into an equivalent RSAS-locked circuit.
These steps essentially invert the functionality of each critical minterm at two places: the first at the wire in the original circuit where RSAS is integrated, and the other at the RSAS block's output. Since these two wires are XOR'ed, the two inversions will cancel out which makes the RSAS-locked circuit functionally equivalent to the SAS-locked circuit.
Due to the equivalence in functionality, the derivations of SAS's \textit{SAT resilience} and \textit{effectiveness} will also hold for RSAS. Therefore, Table \ref{tab:configs} is also the summary of these properties of RSAS.

\section{Choosing Critical Minterms}
\label{sec:critical_minterms}

The critical minterms for injecting large errors should be selected judiciously. 
A careful analysis of the workload would help identify these typical minterms. Generally these minterms would be very few as compared to the overall input space of the functional modules. 
Here we describe how to select the critical minterms.
As mentioned in Sec.~\ref{sec:attacks}, we use PARSEC (generic) and ML models as application benchmarks. 
For the PARSEC benchmarks, we arbitrarily choose critical minterms from the input minterms that exist in all the application benchmarks.
We take a similar approach for ML benchmarks. 
A significant part of an ML-based application is the parameters of the ML model and it turns out that the parameter values of most ML models follow a similar distribution.
For example, Figure~\ref{fig:weights} shows the distribution of parameters of the LeNet (MNIST dataset) and CaffeNet (ISLVRC-2012 dataset) models. These two are the smallest benchmark and the largest benchmark, respectively.
The parameter distributions are similar across ML benchmarks and many other ML models.
This kind of similarity can be also found among generic applications.

We select a subset of parameter values to be critical minterms based on their application-level impact. The selected critical minterms should cause significant application-level error. Fig.~\ref{fig:weights} also shows the accuracy loss of the ML model in the following experiment: for each input minterm, we measure the accuracy loss of the ML model when every computation involving {this very minterm} is corrupted while {no other} minterm is corrupted.
\begin{figure}[h]
    \centering
    \includegraphics[width=0.49\columnwidth, trim={0cm, 0cm, 0cm, 0cm}, clip]{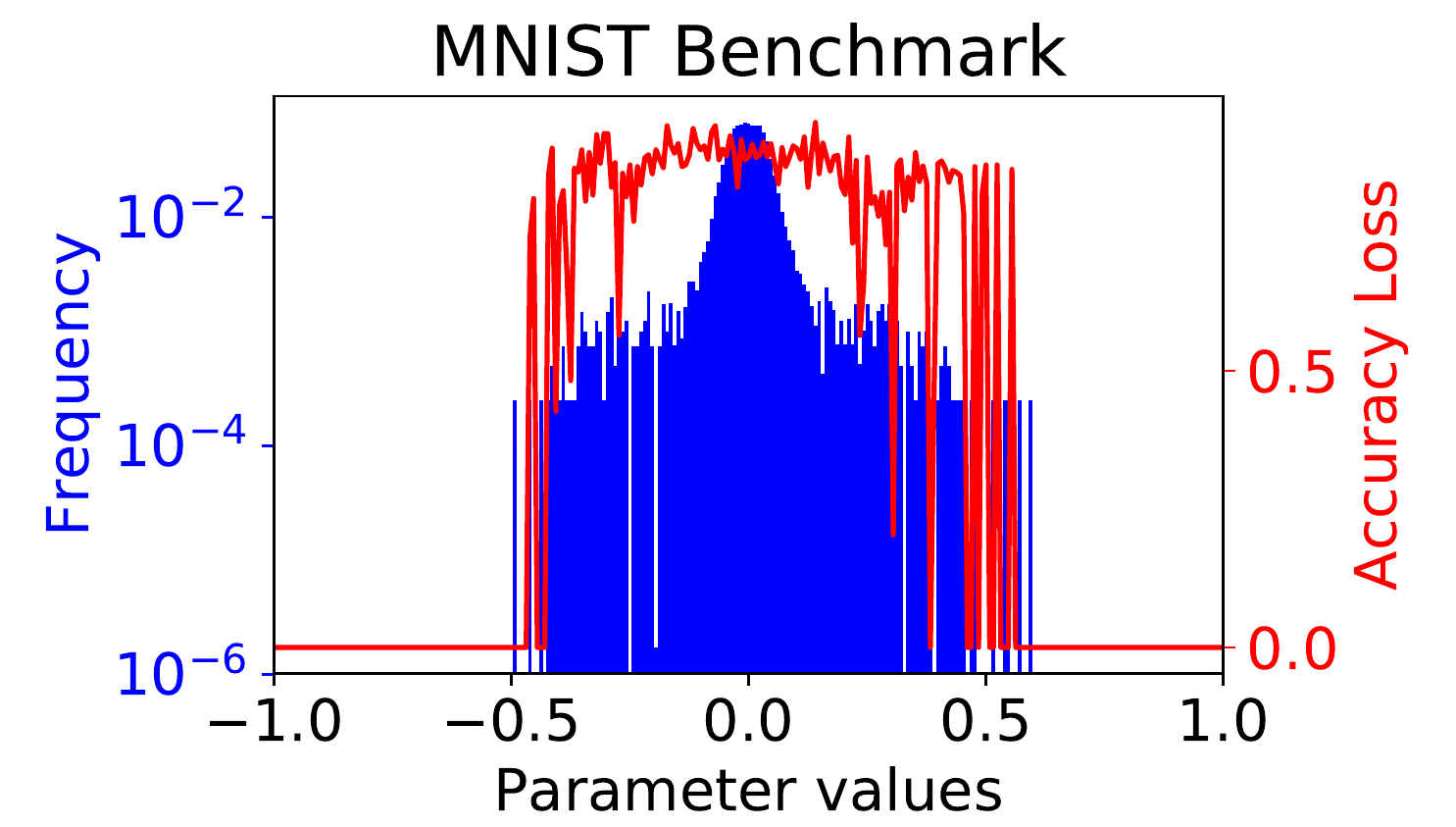}
    \includegraphics[width=0.47\columnwidth, trim={0cm, 0cm, 0cm, 0cm}, clip]{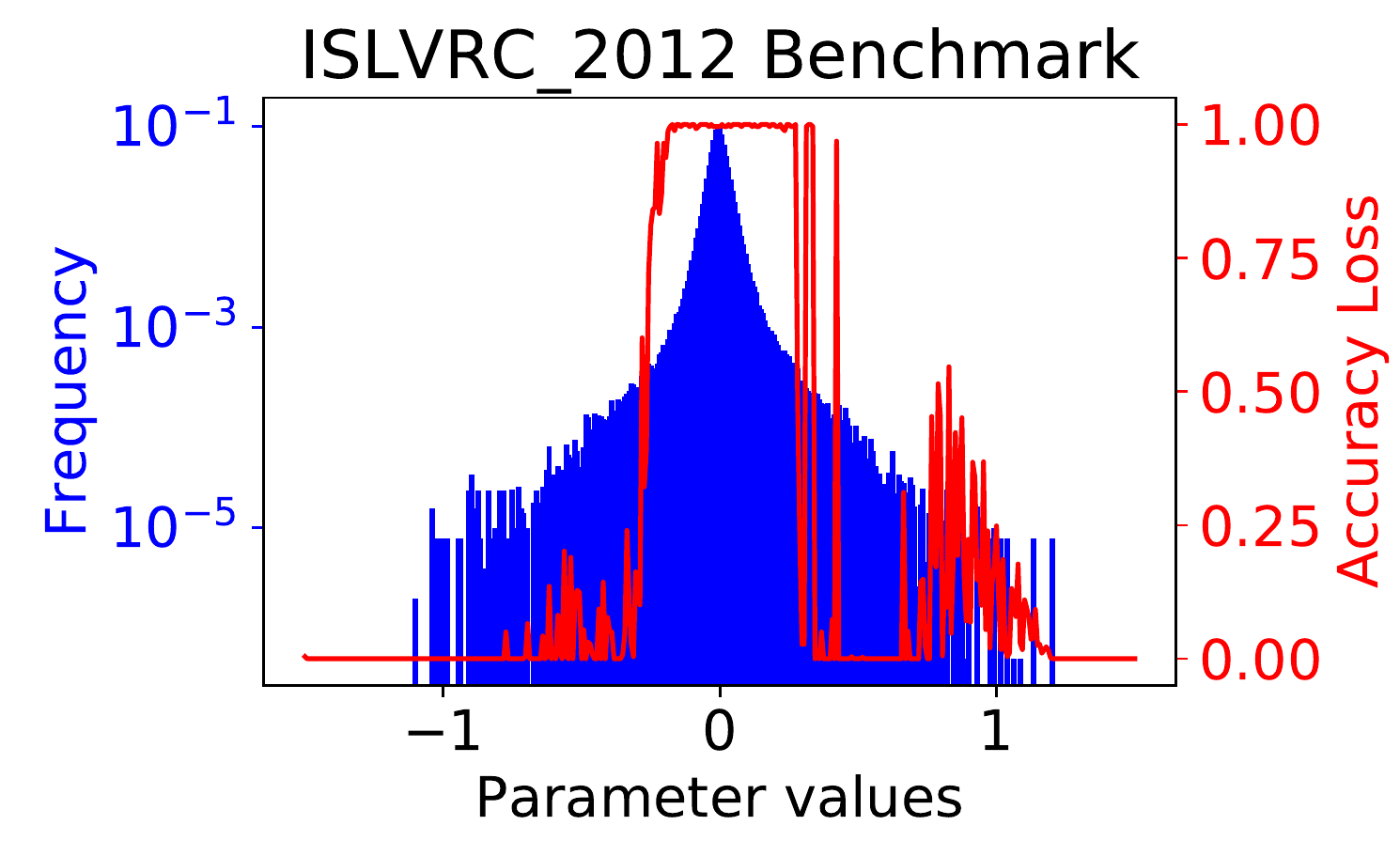}
    \caption{Weight Distribution (Blue Histogram, Left Y Axis) and Application-Level Accuracy Loss (Red Line, Right Y Axis) of LeNet and CaffeNet when the corresponding inputs were locked}
    \label{fig:weights}
\end{figure}
As the input minterm distributions are similar among the same type of applications, the flexibility of SAS/RSAS allows the designer to choose a configuration and a combination of critical minterms that work well in securing the intended applications without compromising SAT resilience. 

{It should be assumed that the application programs (written in binary code) are public knowledge according to the Kerckhoff's principle. An attacker may be tempted to take advantage of the above-mentioned strategy to identify the critical minterms. Since the binary code does not contain the locked module's input values directly, the attacker needs to observe the inputs of the locked module from the scan chain. There are scan chain obfuscation techniques which will corrupt the scan output unless the correct scan chain key is provided \cite{rahman2019dynamically, zhang2017dynamically, karmakar2020securing}. Alternatively, the designer can develop software-based defenses to detect malicious scan chain access or simply burn the scan chain before deploying the activated chips to the open market. Therefore, we do not consider this kind of attack as an immediate security threat to our technique.}

\section{Experiments \& Comparison with SFLL}\label{Sec_Exp}
This section shows the experimental results of SAS and RSAS as well as the comparison with SFLL. 
Recall that, as illustrated in Fig.~\ref{fig:error_transfer}, we obtain the gate-level netlists of the multiplier within a 32-bit 80386 processor by synthesizing the high-level description using Cadence RTL Compiler.
Then we lock the netlist using various SAS and RSAS configurations and SFLL-flex with the same set of critical minterms. 
Note that the critical minterms are selected according to the method described in Sec. \ref{sec:critical_minterms}.
The architecture-level simulation is conducted by a modified GEM5~\cite{binkert2011gem5} simulator where error is injected into the locked processor module according to the hardware error profile due to the wrong key.
We conduct the following experiment to verify the SAT resilience and effectiveness of SAS and RSAS and compare them with SFLL.

\subsection{SAT Resilience}
We first verify whether the SAT resilience of SAS/RSAS (\ie the actual number of SAT iterations) matches what we have derived in Sec. \ref{sec:countermeasure}. The SAT resilience of SAS/RSAS and SFLL is also compared.
We lock the multiplier in a 32-bit 80386 processor with SAS and RSAS as well as SFLL.
Fig. \ref{fig:rsas_iters} shows the actual and expected number of SAT iterations of multipliers locked with SAS and RSAS.  These numbers are compared to the actual number of iterations of SFLL. In these locking configurations, we use 14 bits of primary input for locking purposes ($n=14$) and experiment with each feasible configuration with up to 4 critical minterms. We can observe that SAS and RSAS have similar numbers of actual SAT iterations and they are both close to the expected value. When there is more than one critical minterms, SAS and RSAS exhibit higher SAT resilience than SFLL. 
This is because the corruptibility of each critical minterm in SFLL is almost 1 no matter how many critical minterms there are. This compromises its SAT resilience.
\begin{figure}[htb]
    \centering
    \includegraphics[width=.8\textwidth]{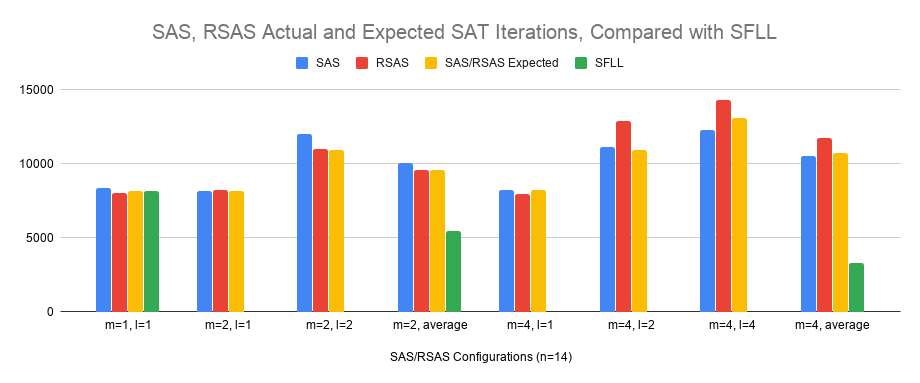}
    \caption{Actual and expected number of SAT iterations of SAS and RSAS, compared with SFLL.}
    \label{fig:rsas_iters}
\end{figure}

Fig.~\ref{fig:sas_vs_sfll_iter} compares the actual SAT iterations of SAS and SFLL. 
In Fig.~\ref{fig:iters_vary_n}, it can be observed that SAS's SAT complexity is higher than that of SFLL by a roughly constant factor when $m$ is fixed at 4. Note that the same set of four critical minterms are used for each locking scheme. Among various SAS configurations, a larger $l$ comes with higher SAT resilience as expected. 
In Fig.~\ref{fig:iters_vary_m}, we vary the critical minterm count ($m$) from 4 to 32 and demonstrate its impact on the SAT resilience of SAS and SFLL. While SAS configurations become stronger with more critical minterms, SFLL becomes weaker. Therefore, SAS is more SAT resilient and gives designers more flexibility when more critical minterms are needed.
\begin{figure}[htb]
    \centering
    \subfloat[Varying key length ($n$), fixing \# critical minterms $m=4$]{\includegraphics[width=0.59\textwidth]{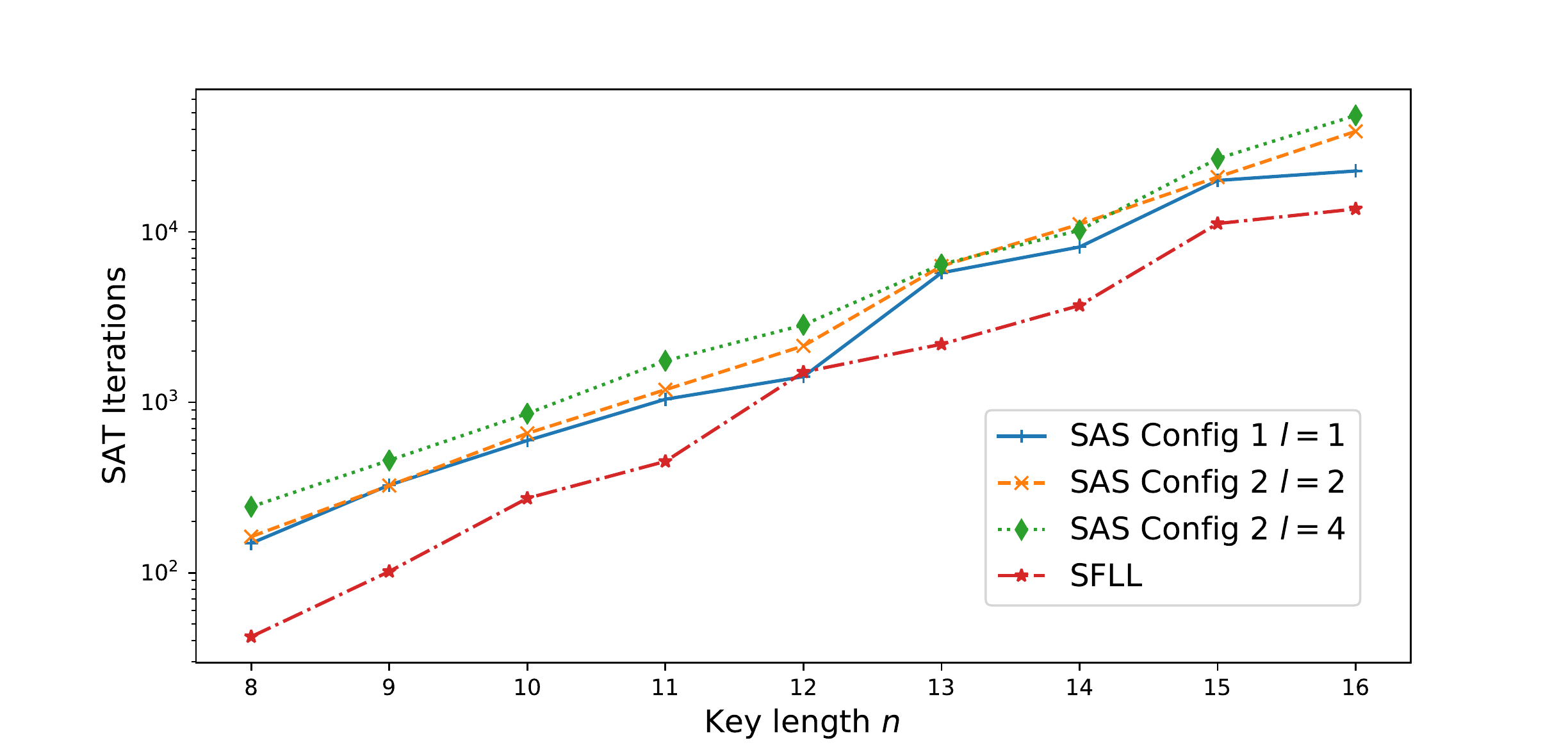}\label{fig:iters_vary_n}}
    \subfloat[Varying \# critical minterms ($m$), fixing key length $n=16$]{\includegraphics[width=0.41\textwidth]{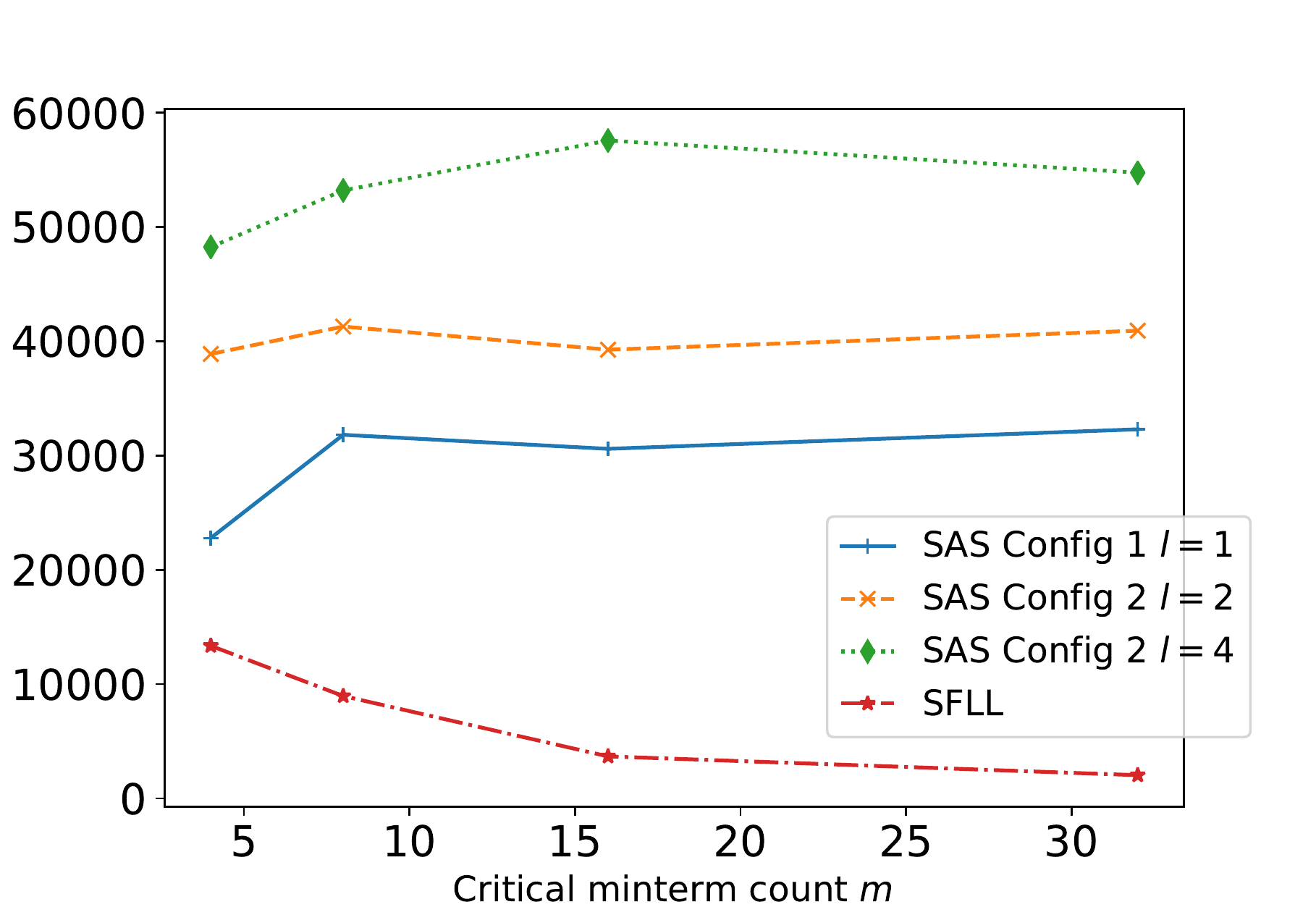}\label{fig:iters_vary_m}}
    \caption{The observed SAT iterations of SAS and SFLL by varying key length and critical minterm count.}
    \label{fig:sas_vs_sfll_iter}
\end{figure}

\subsection{Effectiveness}
We evaluate the effectiveness of SAS/RSAS and SFLL at the application level using PARSEC \cite{bienia2008parsec} and ML benchmarks as listed in Table \ref{tab:dataset_network}. Due to the functional equivalence of SAS and RSAS, they will have the same architecture-level effects and we use the same functional model to perform architecture-level simulation of SAS and RSAS. In our experiments, various numbers of critical minterms are locked. The same set of critical minterms are used for SAS/RSAS and SFLL in each experiment. 
The critical minterms are chosen according to the methods described in Sec.~\ref{sec:critical_minterms}. For SAS, we choose $l=1$ when $m=1$ and $l=2$ when $m\ge 2$. Figs.~\ref{fig:app_eff_parsec} and \ref{fig:app_eff_dnn} show that both SAS/RSAS and SFLL are effective at the application level for both generic and ML-based applications.
SAS/RSAS achieves high application-level effectiveness and exponential SAT resiliency at the same time. 
Considering that SAS/RSAS's SAT resilience is not compromised with the increase in $m$ as opposed to SFLL (as shown in Figs. \ref{fig:rsas_iters} and \ref{fig:iters_vary_m}), SAS/RSAS is a significant improvement over SFLL.

\begin{figure}[htb]
    \centering
    \subfloat[SAS/RSAS on PARSEC]{\includegraphics[width=0.4\textwidth]{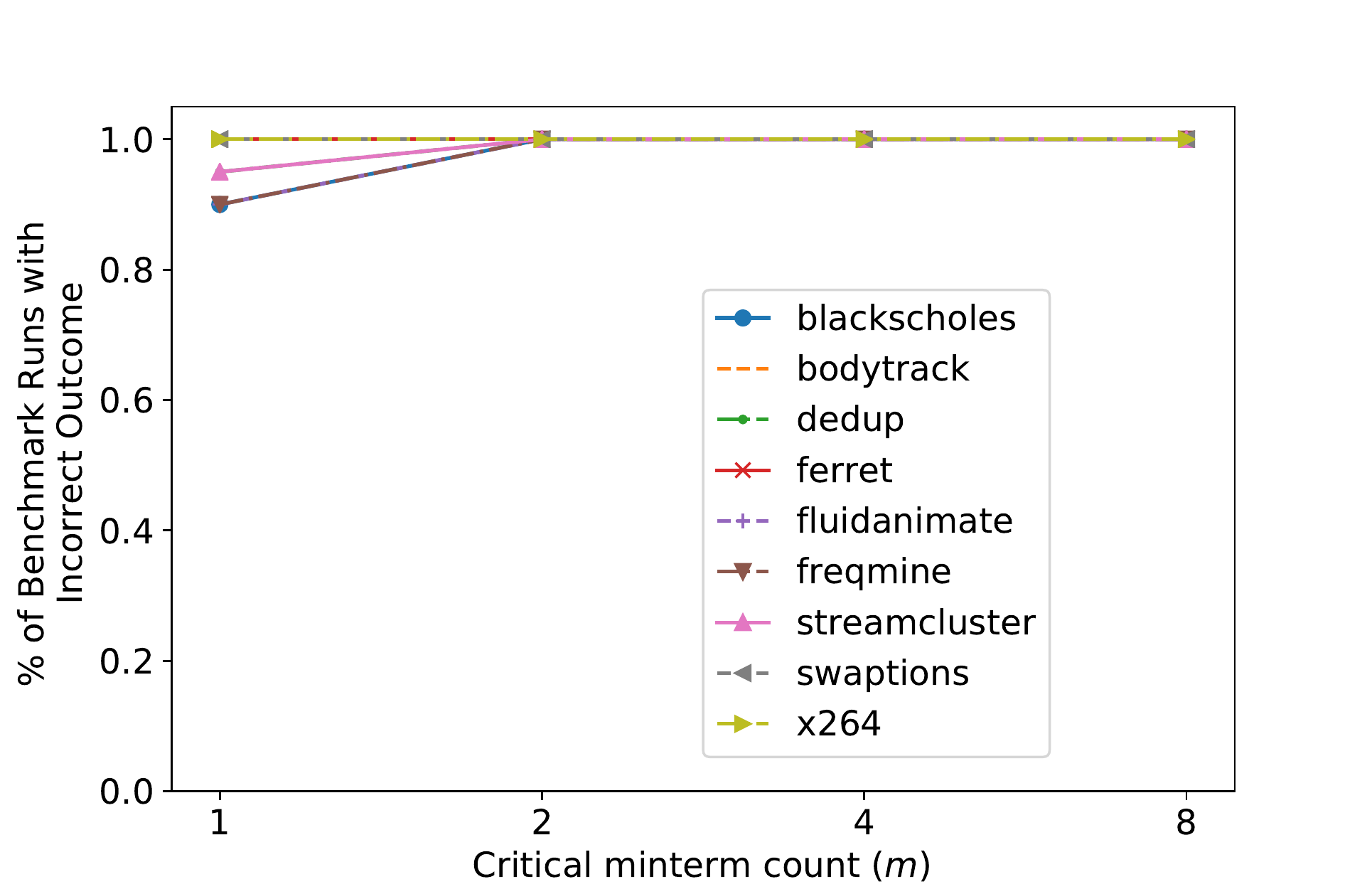}}
    \subfloat[SFLL on PARSEC]{\includegraphics[width=0.4\textwidth]{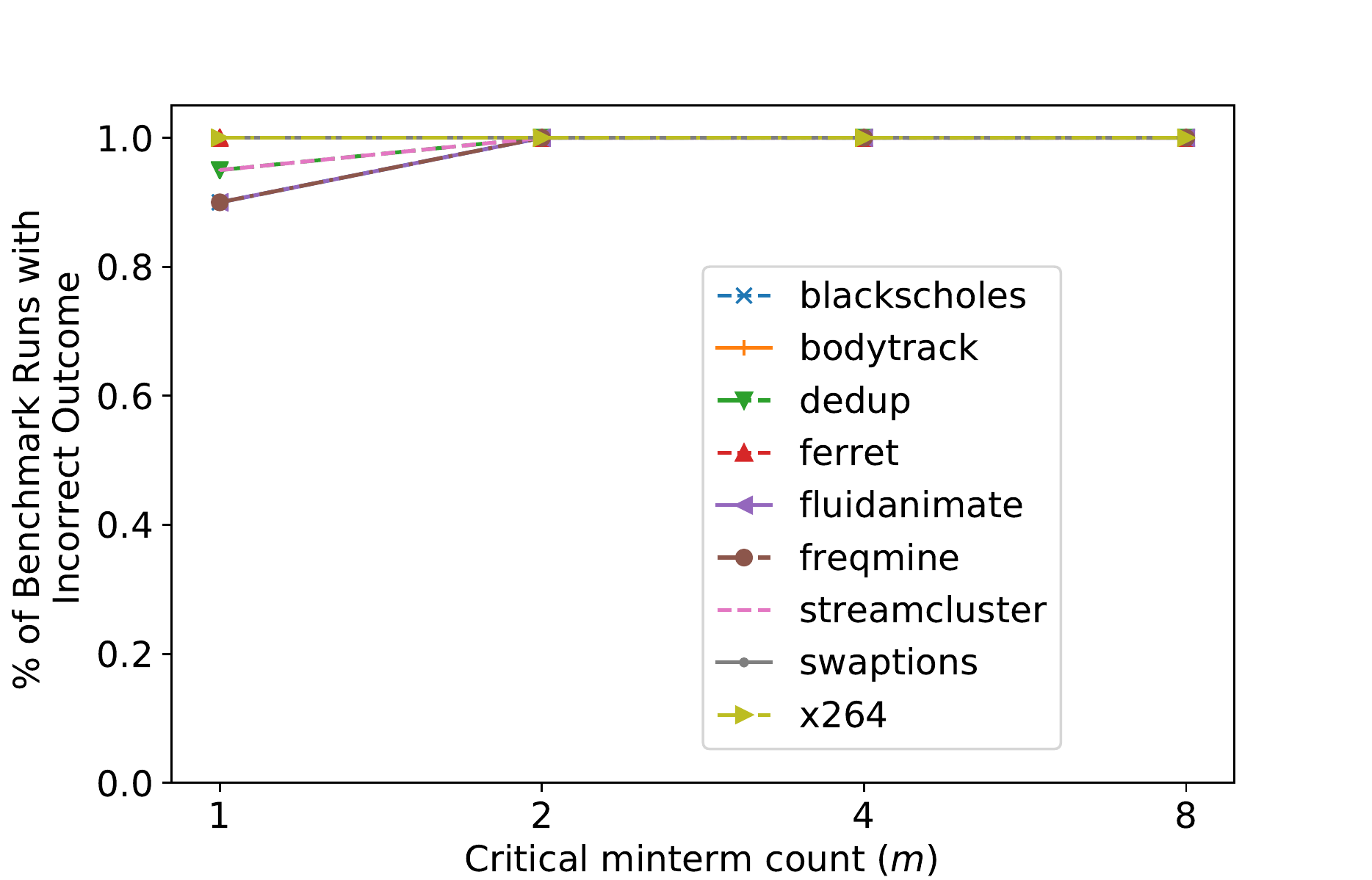}}
    \caption{The application-level effectiveness of SAS/RSAS and SFLL on PARSEC and ML benchmarks}
    \label{fig:app_eff_parsec}
\end{figure}

\begin{figure}[htb]    
    \subfloat[SAS/RSAS on ML]{\includegraphics[width=0.4\textwidth]{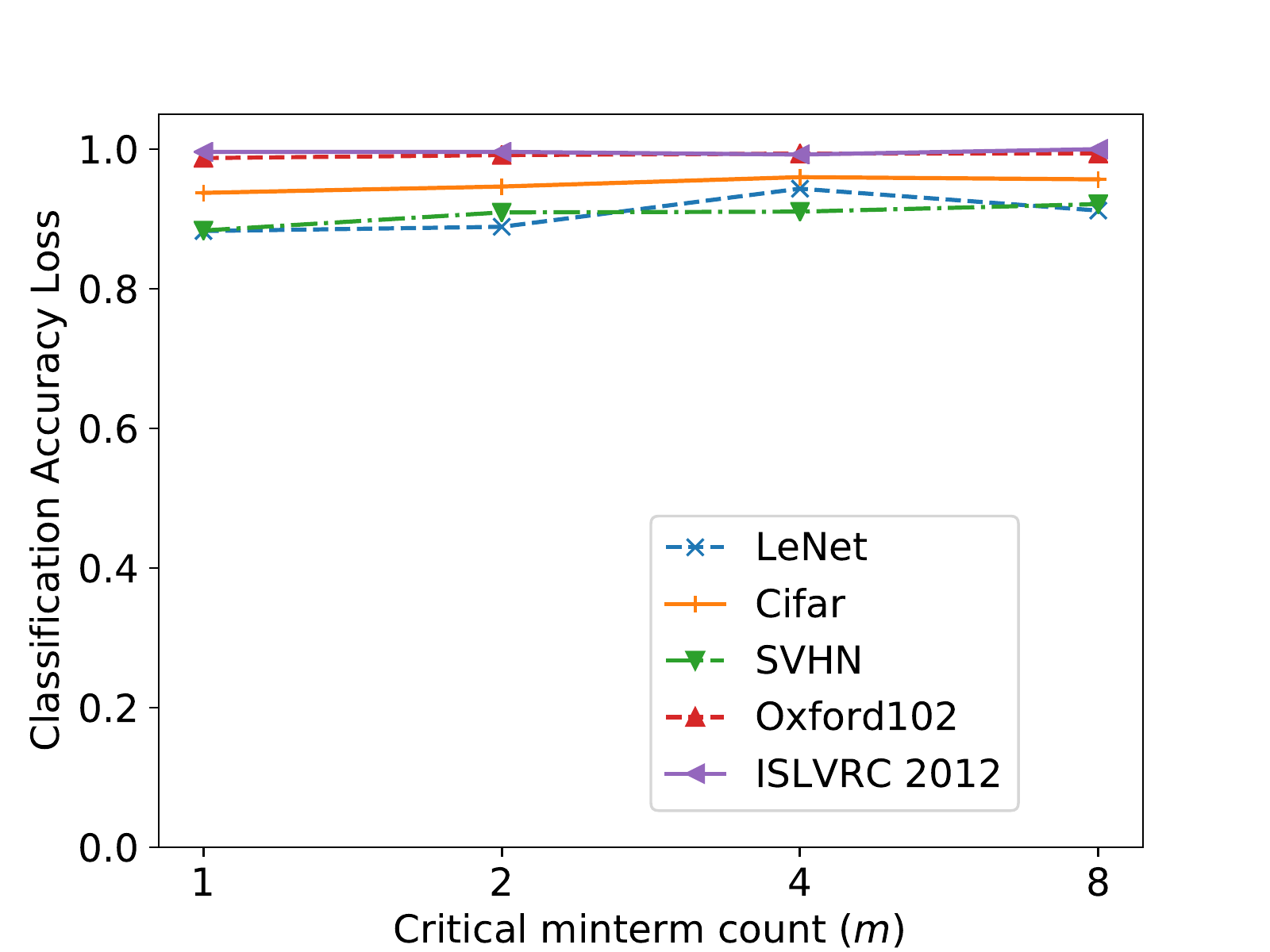}}
    \subfloat[SFLL on ML]{\includegraphics[width=0.4\textwidth]{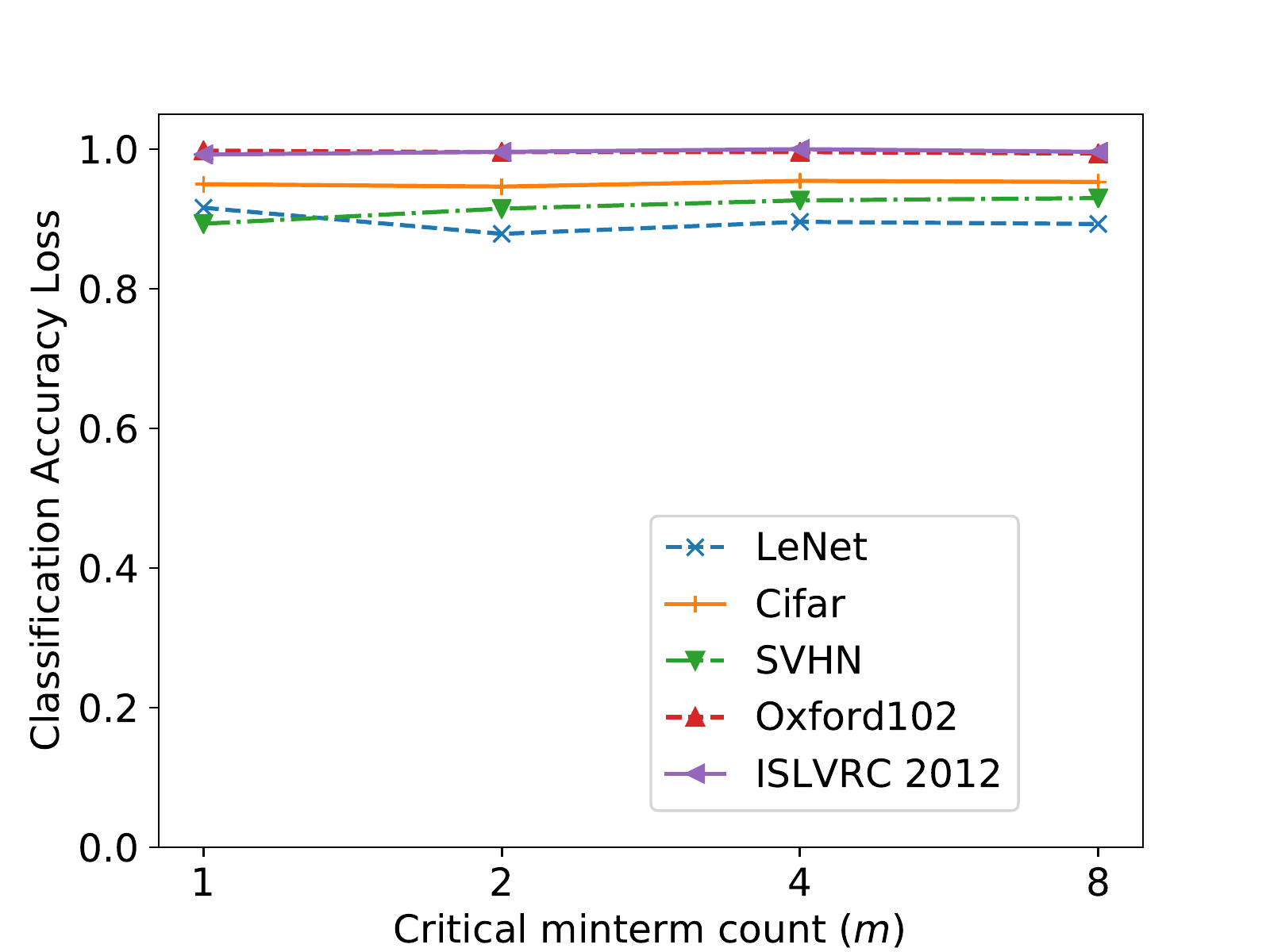}}
    \caption{The application-level effectiveness of SAS/RSAS and SFLL on PARSEC and ML benchmarks}
    \label{fig:app_eff_dnn}
\end{figure}

\subsection{Area, Power, and Delay Overhead of SAS, RSAS, and SFLL}
Now that we have demonstrated the SAT resilience of SAS and RSAS and their application-level effectiveness, we evaluate their area, power, and delay overhead. The overhead is also compared with SFLL. In our evaluation, we use 32 bits from the primary input for locking ($n=32$) and lock up to 4 critical minterms ($m=1,2,4$). We synthesize the original and locked circuits using Cadence RTL Compiler using SAED 90nm process. 
{In order to account for the area of tamper-proof memory that store keys of SAS and look-up tables of SFLL, we add 2.80$\mu m^2$ for each memory bit to the area of locked designs. This is scaled from \cite{yasin2017provably,sengupta2020truly} which used a 65nm library and reported $1.46\mu m^2$ per bit of tamper-proof memory for the square of the feature size, \ie $1.46\mu m^2 \times (\frac{90nm}{65nm})^2 = 2.80\mu m^2$.} Figs. \ref{fig:area_overhead}, \ref{fig:power_overhead}, and \ref{fig:delay_overhead} show the area, power, and delay overhead values, respectively. 
Compared with SFLL, on average, SAS and RSAS have 1.90\% and 0.73\% more area overhead, 0.43\% more and 0.04\% less power overhead, 0.93\% and 0.71\% more delay overhead, respectively. 
These are not significant increases in overhead and should be worth the gain in SAT resilience.
\begin{figure}[htb]
    \centering
    \includegraphics[width=0.8\textwidth]{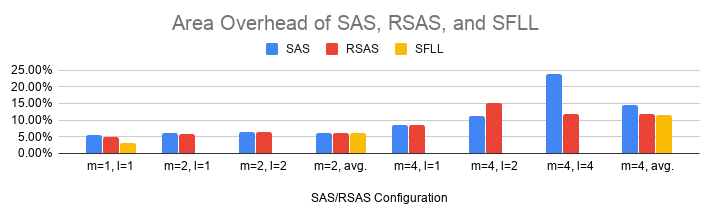}
    \caption{Area overhead of SAS and RSAS compared with SFLL}
    \label{fig:area_overhead}
\end{figure}
\begin{figure}[htb]
    \centering
    \includegraphics[width=0.8\textwidth]{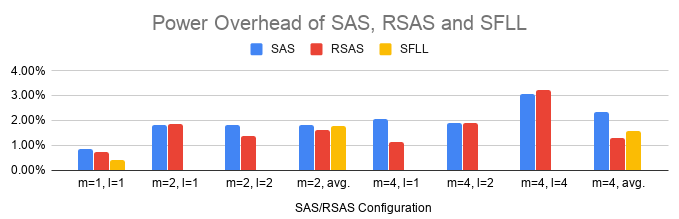}
    \caption{Power overhead of SAS and RSAS compared with SFLL}
    \label{fig:power_overhead}
\end{figure}
\begin{figure}[htb]
    \centering
    \includegraphics[width=0.8\textwidth]{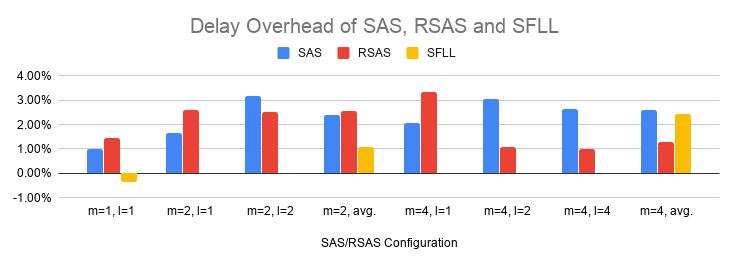}
    \caption{Delay overhead of SAS and RSAS compared with SFLL}
    \label{fig:delay_overhead}
\end{figure}

\section{Conclusion}\label{Sec_Conclusion}
In this work, we investigate logic locking techniques to secure both generic and error-resilient workloads running on locked processors. 
We motivate our work by demonstrating the insufficiency of the state-of-the-art logic locking scheme in securing such applications. We point out that this is due to the fundamental trade-off between {\it SAT resilience} (SAT attack complexity) and {\it effectiveness} (error rate of wrong keys) of logic locking. We formally prove this trade-off.
In order to address this dilemma, we propose Strong Anti-SAT (SAS) where a set of critical minterms are assigned higher corruptibility in order to ensure high application-level impact. Based on SAS, we also propose Robust SAS (RSAS) to thwart removal attacks on logic locking. RSAS is functionally equivalent to SAS and has the same SAT resilience and effectiveness.
Experimental results show that SAS and RSAS secure processors against SAT attack by ensuring exponential SAT attack complexity and high application-level impact simultaneously given any wrong key.
We also evaluate the area, power, and delay overhead of SAS and RSAS and compare it with SFLL. It is shown that SAS and RSAS have modest increase in overhead.
In summary, RSAS exhibits a higher SAT resilience than SFLL when multiple critical minterms are secured, while also maintaining equivalent effectiveness and removal attack resilience. Therefore, RSAS constitutes a significant improvment over SFLL-based locking.

%

\begin{acks}
This work is supported by AFOSR MURI under Grant FA9550-14-1-0351 and Northrop Grumman Corporation and University of Maryland Seedling Grant.
\end{acks}

%
\bibliographystyle{ACM-Reference-Format}
\bibliography{ref}

%
\appendix

\section{Proof of Theorem~\ref{thm:SAT}}\label{proof:thm_sat}
\begin{proof}
This can be proved by contradiction: suppose the key returned by the last step of SAT attack is a wrong key.
This implies that there must exist an input minterm $\vec{X}$ such that
\begin{equation*}
C(\vec{X},\vec{K}_c,\vec{Y}_c)\land C(\vec{X},\vec{K},\vec{Y}) \land (\vec{Y}_c\ne \vec{Y})
\end{equation*}
where $\vec{K}$ is the actual key, $\vec{Y}_c$ is the output with returned key $\vec{K}_c$ and $\vec{Y}$ is the correct output according to the actual key $\vec{K}$.
$\vec{X}$ cannot be a previously found DI because otherwise $\vec{K}_c$ will not satisfy \eqref{SAT_result}.
We can see that $\vec{X}$ qualifies for a DI: just assign $\vec{K}_\alpha = \vec{K}_c$ and $\vec{K}_\beta = \vec{K}$.
This means that~\eqref{SAT_formula} is still satisfiable and contradicts the criteria that no more DI can be found before the SAT attack goes to the final step.

Hence proved.
\end{proof}

\section{Proof of Theorem~\ref{thm:equivalence}}\label{proof:thm_equi}
\begin{proof}
Recall that
\begin{equation*}
    e_w = \frac{1}{|\KKK^W|}\sum_{\forall\vec{K}\in\KKK^W}e_{\vec{K}} = \frac{1}{|\KKK^W|}\sum_{\forall\vec{K}\in\KKK^W} \frac{|\XXX_{\vec{K}}|}{2^n} = \frac{1}{2^n|\KKK^W|}\sum_{\forall\vec{K}\in\KKK^W} |\XXX_{\vec{K}}|
\end{equation*}
and
\begin{equation*}
    \gamma = \frac{1}{2^n}\sum_{\forall\vec{X}\in\{0,1\}^n}\gamma_{\vec{X}} = \frac{1}{2^n}\sum_{\forall\vec{X}\in\{0,1\}^n} \frac{|\KKK_{\vec{X}}|}{|\KKK^W|} = \frac{1}{2^n|\KKK^W|}\sum_{\forall\vec{X}\in\{0,1\}^n} |\KKK_{\vec{X}}|
\end{equation*}
Therefore, in order to prove $e_w=\gamma$, we only need to prove
\begin{equation}
 \sum_{\forall\vec{K}\in\KKK^W} |\XXX_{\vec{K}}| = \sum_{\forall\vec{X}\in\{0,1\}^n} |\KKK_{\vec{X}}|  
 \label{eq:equalityproof}
\end{equation}
Let us consider the following bipartite graph $G=(\XXX, \KKK^W, \EEE)$ where $\XXX$ is $\{0,1\}^n$ which is the set of all the possible input minterms, $\KKK^W$ is the set of wrong keys, and the set of edges $\EEE=\{(\vec{X},\vec{K})|\vec{X}\in\XXX\text{ and }\vec{K}\in\KKK^W\text{, }\vec{K}\text{ corrupts }\vec{X}\}$.
Both sides of Eq.~\ref{eq:equalityproof} denote the total number of edges in $\EEE$ and hence must be equal.
\end{proof}

\color{black}

\section{Proof of Lemma~\ref{lemma:removal}}\label{proof:lemma_removal}
\begin{proof}
Recall that Equation~\eqref{SAT_formula} gives the SAT formula for each SAT iteration:
        \begin{equation*}
        \begin{aligned}
        C(\vec{X}_i,\vec{K}_\alpha,\vec{Y}_\alpha)\land C(\vec{X}_i,\vec{K}_\beta,\vec{Y}_\beta)\land (\vec{Y}_\alpha \ne \vec{Y}_\beta) \\
        \bigwedge_{j=1}^{i-1}(C(\vec{X}_j,\vec{K}_\alpha,\vec{Y}_j)\land C(\vec{X}_j,\vec{K}_\beta,\vec{Y}_j))
        \end{aligned}
        \end{equation*}
To satisfy the first line, at least one of $\vec{K}_\alpha$ and $\vec{K}_\beta$ must be a wrong key that corrupts $\vec{X}$.
However, if any wrong key that corrupts $\vec{X}_i$ also corrupts at least 1 previously found DI, this wrong key cannot satisfy the second line.
Therefore, such $\vec{X}_i$ cannot be the DI in future iterations.
\end{proof}

\section{Proof of Lemma~\ref{lemma:critical_count}}\label{proof:lemma_critical_count}
\begin{proof}
Recall that $g$ has on-set size 1.
Let $\vec{X}_g$ be the very input that makes $g(\vec{X}_g)=1$.
$\forall \vec{X}\in \MMM$, let $\vec{K_1}=\vec{X}\oplus\vec{X}_g$.
Then, any $\vec{K}=(\vec{K_1},\vec{K_2})\in \KKK^W$ is a wrong key that only corrupts $\vec{X}$.
Therefore, $\vec{X}$ has to be chosen as a DI to prune out this wrong key.
\end{proof}

\section{Proof of Lemma~\ref{lemma:SAS_critical_count}}
\label{proof:lemma_sas_crit}
\begin{proof}
This is a natural extension to Lemma~\ref{lemma:critical_count}.
Let $\vec{X}$ be a critical minterm and $\vec{X}\in\MMM^j$.
Recall that $g$ has on-set size 1.
Let $\vec{X}_g$ be the very input that makes $g(\vec{X}_g)=1$.
$\forall \vec{X}\in \MMM^j$, let $\vec{k}=\vec{X}\oplus\vec{X}_g$.
Then, let us consider the following wrong key $\vec{K}=(\vec{K}^1,\vec{K}^2,\ldots,\vec{K}^l)\in \KKK^W$ which is composed as follows:
$\vec{K}^{j}=(\vec{k},\vec{K}^j_2)\in\KKK^W_j$ where $\KKK^W_j$ is the set of wrong keys for the $j^\text{th}$ SAS block.
For any $i=1,2,\ldots,l$ that $i\ne j$, $\vec{K}^i \in \KKK^C_i$ where $\KKK^C_i$ is the set of correct keys for the $i^\text{th}$ SAS block.
Such a key $\vec{K}$ is a wrong key that only corrupts $\vec{X}$.
Therefore, $\vec{X}$ has to be chosen as a DI to prune out this wrong key.
\end{proof}
\end{document}